\documentclass[11pt,letterpaper]{article}
\usepackage[margin=1in]{geometry}
\usepackage{authblk}

\usepackage{graphicx}
\usepackage{amsfonts,amsmath,amssymb,amsthm,mathtools}
\usepackage{comment}
\usepackage{subcaption}
\usepackage{xcolor}
\definecolor{ForestGreen}{rgb}{0.1333,0.5451,0.1333}
\definecolor{DarkRed}{rgb}{0.80,0,0}
\definecolor{Red}{rgb}{1,0,0}
\usepackage[linktocpage=true,
pagebackref=true,colorlinks,
linkcolor=DarkRed,citecolor=ForestGreen,
bookmarks,bookmarksopen,bookmarksnumbered]
{hyperref}
\usepackage{cleveref}
\usepackage{algorithm}
\usepackage[noend]{algpseudocode}
\usepackage{array}
\usepackage{booktabs}
\usepackage{multirow}
\usepackage{tikz}
\usetikzlibrary{arrows.meta,shapes.geometric}

\DeclareMathOperator*{\Var}{Var}
\DeclareMathOperator{\med}{Med}
\DeclareMathOperator{\Median}{\med}
\DeclareMathOperator{\Medianp}{Median}
\DeclareMathOperator{\anchor}{Anchor}
\DeclareMathOperator{\bargain}{Bargain}

\theoremstyle{plain}
\newtheorem{theorem}{Theorem}[section]
\newtheorem{lemma}[theorem]{Lemma}
\newtheorem{proposition}[theorem]{Proposition}
\newtheorem{corollary}[theorem]{Corollary}

\theoremstyle{definition}
\newtheorem{definition}[theorem]{Definition}

\newtheorem{example}[theorem]{Example}

\newcommand{\sbpara}[1]{{\smallskip\noindent\textbf{#1}}}
\newcommand{\Exp}{\mathbb{E}}
\newcommand{\abs}[1]{\left| #1 \right|}
\newcommand{\ind}{\mathbf{1}}

\newcommand{\Real}{\mathbb{R}}
\newcommand{\Qcal}{\mathcal{Q}}
\newcommand{\Util}{\widetilde{U}}
\newcommand{\Vtil}{\widetilde{V}}
\newcommand{\Utilde}{\widetilde{U}}
\newcommand{\Vtilde}{\widetilde{V}}
\newcommand{\SCost}{\mathsf{SC}}
\newcommand{\Distortion}{\mathsf{Distortion}}

\begin{document}

\title{Anchored Sequential Deliberation}

\author[1]{Sijing Tu~\footnote{sijingtu@stanford.edu}}
\author[1]{Ashish Goel~\footnote{ashishg@stanford.edu}}
\affil[1]{%
  Management Science and Engineering, Stanford University
}
\date{}

\maketitle

\begin{abstract}
Sequential deliberation is a mechanism for collective decision making:
at each round, a uniformly randomly selected pair is asked to revise a collective outcome, which then becomes the input for the next round.
Existing theory by Fain et al.~\cite{fain2017sequential} treats the current outcome solely as the disagreement alternative in bargaining.
Yet the existing outcome might also carry social influence and anchor participants' positions toward the status quo.
In this paper, we introduce \emph{anchored sequential deliberation}.
At each round, two participants with bliss points $U$ and $V$ shift their positions toward the previous outcome $O_{t-1}$ with anchoring strength $0\leq \lambda<1$.
They then bargain using $O_{t-1}$ as the disagreement alternative.
For the sake of analysis, we assume that the decision space is one-dimensional, the anchoring effect is linear, and participants use Nash bargaining.
The update simplifies to
$O_t=(1-\lambda)\Medianp\{U,V,O_{t-1}\}+\lambda O_{t-1}$.

Our analysis reveals a trade-off.
For every population distribution, the outcome has a unique stationary distribution.
Through a coupling of two outcomes, we find that the process contracts in $1$-Wasserstein distance with a factor of at most $\frac{1+\lambda}{2}$, which implies that stronger anchoring slows mixing.
On the other hand, stationary distortion weakly decreases with $\lambda$, although the worst-case distortion remains $\frac{1+\sqrt{2}}{2}$ for every feasible $\lambda$.
We also identify a unique \emph{deliberative fixed point}, at which the expected movement is zero, and prove that the stationary distribution concentrates around it as $\lambda$ increases.
For symmetric populations, we further provide a tighter bound on the stationary variance around this fixed point.
For the uniform distribution, we establish upper and lower bounds on stationary distortion, both of which approach $1$ as $\lambda$ approaches $1$.
Simulations for uniform and Beta populations show that stronger anchoring slows mixing, concentrates the stationary distribution, and lowers stationary distortion in these instances.
\end{abstract}

\section{Introduction}
\label{sec:intro}
Public collective decision making has recently attracted growing attention in the computer science literature. 
A significant line of work models collective choice as a social choice problem, where participants' preferences are aggregated through a voting rule~\cite{charikar2024breaking,anshelevich2018approximating,gkatzelis2020resolving,kizilkaya2022plurality,kizilkaya2023generalized,munagala2019improved}.
In many public decision-making processes, however, voting is only the final step and should not replace deliberation~\cite{bachtigerOxfordHandbookDeliberative2018}.
Deliberation helps participants become better informed about the issues and can improve the quality of collective decisions.
This importance has motivated practical platforms for public deliberation, including the Online Deliberation Platform~\cite{onlineDeliberationPlatform}, Polis~\cite{polis}, and Decidim~\cite{decidim}.

A practical mechanism for deliberation should be simple to implement, decentralized, scalable to large populations, and should not require participants to reason over the entire decision space~\cite{fain2017sequential}.
Motivated by these principles, recent work has proposed mechanisms based on small-group deliberation~\cite{goelLargeScaleDeliberativeDecisionMaking2016,fain2017sequential,goel2025metric}.
These mechanisms are evaluated using distortion, defined as the ratio between the social cost of the chosen outcome and the optimal social cost~\cite{procaccia2006distortion}.
This line of research shows that integrating deliberation into a voting rule can produce higher-quality outcomes with lower distortion.

In this paper, we build on the sequential deliberation framework of Fain et al.~\cite{fain2017sequential}: in each round, a pair of participants is drawn from the population to bargain over the deliberative outcome from the previous round, which serves as the disagreement alternative.
We add an individual-behavior perspective that is natural in real-world deliberation but absent from the established theoretical model.
When participants deliberate over an existing policy, draft proposal, or budgeting arrangement, they might not interpret the current outcome only as a disagreement alternative. The current outcome can also act as a source of social influence that shapes what the participants are willing to propose, as changing an existing outcome can require effort, justification, and social risk.
We model this tendency to stay close to the current outcome as an \emph{anchoring bias}~\cite{tversky1974judgment}.
We use a research lab writing a shared lab rule as an illustrative example.

\begin{example}
Assume a research lab wants to write a shared lab rule. The professor writes the first draft. In each round, two randomly selected students revise the current draft by mutual agreement. Each student may have an ideal version of the rule: some prefer more supervision, while others prefer working independently. Yet large edits can be effortful or confrontational, especially when the current text already looks legitimate. As a result, the pair might propose an incremental revision that reflects their most important shared concerns, rather than an outcome obtained by bargaining directly from their ideal positions.
\end{example}

Thus, a single round of deliberation has two components: anchoring pulls participants' expressed stances toward the previous outcome, and bargaining turns those expressed stances into the next outcome.
We now introduce a model of anchored sequential deliberation that incorporates both components.

\floatname{algorithm}{Process}
\begin{algorithm}[t]
\caption{Sequential Deliberation Process}
\label{alg:anchored-sequential-deliberation}
\begin{algorithmic}[1]
    \Require
    Participants provide their preferred outcomes; an arbitrary outcome is selected as $O_0$.
    \For{$t = 1, 2, \ldots, T$}
        \State Select uniformly at random a pair of participants $u$ and $v$ with bliss points $U$ and $V$, respectively.
        \State Provide $u$ and $v$ with the current deliberation outcome $O_{t-1}$.
        \State The participants form anchored stances:
        \begin{equation}
        \label{eq:anchor-stance}
        \Util = \anchor_{\lambda} (U, O_{t-1}) \qquad
        \Vtil = \anchor_{\lambda} (V, O_{t-1})
        \end{equation}
        \State Participants $u$ and $v$ bargain using $O_{t-1}$ as the disagreement alternative.
        \begin{equation}
        \label{eq:bargain-step}
        O_{t} = \bargain(\Util, \Vtil, O_{t-1}).
        \end{equation}
    \EndFor
    \State \Return $O_T$
\end{algorithmic}
\end{algorithm}

\subsection{Our model: anchored sequential deliberation}
We formulate anchored sequential deliberation based on the framework of Fain et al.~\cite{fain2017sequential}; the process is shown in Process~\ref{alg:anchored-sequential-deliberation}.
In each round, the participants first observe the outcome from the previous round and shift their stances toward it, as in Equation~\ref{eq:anchor-stance};
they then bargain using that previous outcome as the disagreement alternative, as in Equation~\ref{eq:bargain-step}.

The function $\anchor_{\lambda}(\cdot)$ takes a participant's ideal outcome (i.e., bliss point) and the previous outcome $O_{t-1}$ as inputs and returns the participant's anchored stance.
The function $\bargain(\cdot)$ then takes the anchored stances and the previous deliberation outcome and returns the deliberative outcome for the current round.

\subsection{Analytical assumptions}
\label{sec:analyticalassumptions}
For the sake of analysis, we assume that the deliberation space is the interval $[0,1]$.
Each agent has a bliss point, representing the agent's ideal outcome, drawn from a cumulative distribution function $G$ on $[0,1]$.

\sbpara{The anchoring parameter $\lambda$.} Let $\lambda \in [0,1)$ measure the strength of anchoring: the extent to which participants move their expressed stances toward the previous deliberative outcome before bargaining.
Larger $\lambda$ means stronger anchoring, so expressed stances lie closer to the current outcome.
We use the following simple linear update rule:

\begin{equation}
    \label{eq:anchoring}
    \Util = (1-\lambda) U + \lambda O_{t-1},
    \qquad
    \Vtil = (1-\lambda) V + \lambda O_{t-1}
\end{equation}

This linear update is motivated by the empirical study of Lieder et al.~\cite{lieder2018empirical}, which finds that anchoring bias grows linearly with the distance between the anchor and the correct value.
It is also consistent with social-influence models in which people revise their opinions by weighting and combining their own positions with others' positions, including the DeGroot model~\cite{degroot1974reaching} and the Friedkin--Johnsen model~\cite{friedkin1990social}.

\sbpara{Nash bargaining.}
For tractability, we assume that the two selected participants use Nash bargaining.
Given the disagreement alternative $O_{t-1}$ and the anchored stances $\Util$ and $\Vtil$, their utility gains from an outcome $O_t$ are
$d(O_{t-1}, \Util) - d(O_{t}, \Util)$ and $d(O_{t-1}, \Vtil) - d(O_{t}, \Vtil)$, respectively.
The Nash bargaining solution maximizes the Nash product:
\[\arg\max_{O_{t} \in [0, 1]} \left(d(O_{t-1}, \Util) - d(O_{t}, \Util)\right) \left(d(O_{t-1}, \Vtil) - d(O_{t}, \Vtil)\right),\]
where both utility gains are non-negative. 

By Fain et al.~\cite[Lemma~1]{fain2017sequential}, the line is a special case of a median graph, and the Nash bargaining solution is the \emph{median} of the two anchored stances and the disagreement alternative.
We therefore model the sequence of bargaining outcomes as a Markov chain with the following update rule:
\begin{equation}
    \label{eq:bargain-step-outcome}
    O_{t} = \med(\Util, \Vtil, O_{t-1}),
\end{equation}
where the initial outcome $O_0\sim G$ is drawn from the population distribution.

\sbpara{Social cost and distortion.} For an outcome $O$, the social cost is the expected distance from the agents' bliss points to $O$:
$\SCost(O) \coloneqq \int_{0}^{1} |O - x| dG(x)$. Let $O^*$ minimize social cost, so $\SCost(O^*) = \min_{O} \SCost(O)$.
The distortion of $O$ is $\Distortion(O) \coloneqq \frac{\Exp[\SCost(O)]}{\SCost(O^*)}$.

\subsection{Our results}
\label{sec:main-results}
Does anchoring lead to better deliberative results?
Our analysis shows a trade-off between convergence and stability.
Stronger anchoring makes each round stay closer to the previous outcome, so the process takes longer to forget its starting point.
At the same time, once the process reaches its stationary distribution, stronger anchoring can make the outcome more stable.
Here, we define a \emph{deliberative fixed point}, the point where the expected unanchored movement is zero.
We find that increasing the anchoring strength makes the stationary distribution more concentrated around this fixed point.
Moreover, the stationary distortion weakly decreases as the anchor effect strengthens; however,
in the worst case, the distortion remains the same as in the unanchored process.
We therefore study more restricted settings where the deliberative fixed point is the population median, which is socially optimal.
For symmetric populations, stronger anchoring reduces stationary variance around the median; for the uniform case, this gives explicit upper and lower distortion bounds.
We also use simulations on several Beta distributions to demonstrate the same trade-off beyond the cases where we have closed-form bounds.

\sbpara{Convergence.}
For every population distribution $G$ on $[0,1]$ and every $\lambda<1$, the anchored process has a unique stationary distribution.
Its $1$-Wasserstein mixing time satisfies
$
    t_{\mathrm{mix}}^{(W_1)}(\varepsilon)
    \le
    \left\lceil
    \frac{\log(1/\varepsilon)}{\log(2/(1+\lambda))}
    \right\rceil,
$
and, when $0<\lambda<1$, it also has the lower bound
$
    t_{\mathrm{mix}}^{(W_1)}(\varepsilon)
    \ge
    \left\lceil
    \frac{\log(1/(2\varepsilon))}{\log(1/\lambda)}
    \right\rceil
$
for $0<\varepsilon<1/2$ (Theorem~\ref{thm:wasserstein-mixing-time}).
Both bounds increase as $\lambda$ increases.
We also show that the upper-bound rate is worst-case tight up to a factor of $2$ by constructing a symmetric two-cluster population at the endpoints (Lemma~\ref{lem:wasserstein-rate-tight}).
Thus stronger anchoring slows convergence.
The proof uses a monotone coupling.
First, the one-step update can be written as
\begin{equation}
    \label{eq:simplified-update}
    O_t=(1-\lambda)\med(U,V,O_{t-1})+\lambda O_{t-1},
\end{equation}
so the anchored step is a convex combination of the unanchored median step and the previous outcome.
We then run two copies of the chain with the same sampled bliss points in every round.
For any pair of outcomes $X_t$ and $Y_t$ from the two copies with $X_t\le Y_t$, the order is preserved, and the new gap is a linear combination, with weights $\lambda$ and $1-\lambda$, of the old gap $Y_t-X_t$ and the length of the overlap between $[X_t,Y_t]$ and the interval spanned by the two sampled bliss points.
Taking expectations gives
\[
    \Exp[Y_{t+1}-X_{t+1}]
    \leq
    \frac{1+\lambda}{2}\Exp[Y_t-X_t],
\]
which gives the $1$-Wasserstein upper bound and uniqueness of the stationary distribution.
For the lower bound, the same coupling from the endpoints keeps the gap at least $\lambda^t$; Kantorovich--Rubinstein duality (Theorem~\ref{thm:kantorovich-rubinstein}) and the triangle inequality then give the stated mixing-time lower bound.

\sbpara{Stationary distortion bound.}
We next compare the stationary social cost under different anchoring parameters.
For every population distribution $G$, let $0 \leq \lambda_2 \leq \lambda_1 < 1$ and let $X_{i, \infty} \sim \pi_{\lambda_i, G}$ be stationary outcomes.
The expected social cost of the stationary outcome weakly decreases as $\lambda$ increases
(Theorem~\ref{thm:anchored-worst-case}):
\[
\Exp_{X_{1, \infty}\sim\pi_{\lambda_1,G}}[\SCost(X_{1, \infty})]
\le
\Exp_{X_{2, \infty}\sim\pi_{\lambda_2,G}}[\SCost(X_{2, \infty})].
\]
We immediately obtain the same stationary distortion upper bound as for the unanchored process:
$(1+\sqrt{2})/2$.
Actually, this bound is tight: we can construct a two-cluster population with suitable masses that attains the same distortion for all anchoring parameters (Lemma~\ref{lemma:worst-case-tightness}).
The proof of Theorem~\ref{thm:anchored-worst-case} uses two ingredients:
first, we find a convex function $\varphi:[0, 1] \to \Real$ such that
$\Exp[\mathsf{SC}(X_{1, \infty})] - \Exp[\mathsf{SC}(X_{2, \infty})] = \Exp[\varphi(X_{1, \infty})] - \Exp[\varphi(M_1)],$ 
where $U,V\stackrel{\mathrm{i.i.d.}}{\sim}G$ and $M_1 = (1 - \lambda_2) \cdot \text{Med}(X_{1, \infty}, U, V) + \lambda_2 X_{1, \infty}$. 
By the convexity of $\varphi$ and the linearity of expectation, we can show that
$\Exp[\varphi(X_{1, \infty})] - \Exp[\varphi(M_1)] \leq 0.$

\sbpara{Concentration around the deliberative fixed point.}
Every population distribution has a unique \emph{deliberative fixed point} $\theta_G$, defined as the point where the expected unanchored movement is zero (Lemma~\ref{lem:deliberative-fixed-point}).
As $\lambda$ increases, the stationary distribution concentrates around this point (Theorem~\ref{thm:general-stationary-concentration}).
In particular,
\begin{equation*}
    \Exp[(O_\infty-\theta_G)^2] \leq 1-\lambda.
\end{equation*}
The proof uses the drift function $b_G(x)=\Exp[\med(x,U,V)]-x$.
An integral formula for $b_G$ shows that it is strictly decreasing, so it has a unique zero.
Substituting Equation~\eqref{eq:simplified-update} into the stationary second moment $\Exp[(O_\infty-\theta_G)^2]$ and applying the slope bound on $b_G$ gives the bound above.

\sbpara{Symmetric distributions.}
For symmetric population distributions, the deliberative fixed point is the population median, which minimizes social cost on the line.
Hence, concentration indicates a decrease in distortion.
For the uniform distribution, the stationary distortion is at most
$1+\frac{1-\lambda}{6(1+\lambda)}$ and at least
$1+\frac{1-\lambda}{9+7\lambda}$ for every $0\le\lambda<1$ (Theorem~\ref{thm:distortion-uniform}).
Both bounds decrease with $\lambda$ and approach $1$ as $\lambda$ approaches $1$.
When $\lambda=0$ and there is no anchoring effect, the exact stationary distortion is $\pi-2$.
Thus, in the uniform case, stronger anchoring tends to improve the stationary distortion guarantee.
More generally, for centered symmetric distributions,
\[
    \Exp\!\left[\left|O_\infty-\frac{1}{2}\right|^2\right]
    \leq
    \frac{2(1-\lambda)}{1+\lambda}\Var(U)
\]
(Theorem~\ref{thm:centered-second-moment-general-g}).
The proof reduces distortion to variance around the median.
In the uniform case, distortion is exactly $1+4\Exp[(O_\infty-\frac{1}{2})^2]$, and a one-step second-moment recursion gives the upper and lower bounds.
For general centered symmetric distributions, the same recursion is bounded using symmetry and the variance of a fresh bliss point.

\sbpara{Simulations.} We complement the theoretical analysis with simulations for several Beta distributions in Section~\ref{sec:simulations} and Appendix~\ref{app:more_simulations}.
The simulations show the same message: stronger anchoring slows convergence, concentrates the stationary distribution, and lowers stationary distortion in the simulated instances.
Overall, anchoring trades convergence speed for stationary stability: it slows mixing, but it can improve stationary outcome quality for natural distributions.

\subsection{Further related work}
\label{sec:further-literature-review}
\sbpara{Anchoring bias.}
Anchoring was first studied by Tversky and Kahneman~\cite{tversky1974judgment}: when people estimate a quantity after being shown an arbitrary value, their judgments remain biased toward that value.
The effect is robust across domains, tasks, expertise, motivation, and cognitive load~\cite{furnham2011literature}, and arises in both individual~\cite{samuelson1988status} and cooperative group~\cite{wilde2018anchoring} decision making.
Anchors can stem from advice taking~\cite{baileyMetaanalysisWeightAdvice2023a,yanivEgocentricDiscountingReputation2000}, social influence~\cite{degroot1974reaching,friedkin1990social,moussaidSocialInfluenceCollective2013}, or negotiation, where manipulated listing prices cause bias even for expert real-estate agents~\cite{northcraft1987experts} and first offers may shape final agreements~\cite{galinsky2001first}.
Most relevant to our model, Lieder et al.~\cite{lieder2018empirical} empirically show that anchoring bias grows linearly with the distance between the anchor and the correct value, motivating our linear update rule.
For deliberation specifically, Hartmann and Rafiee Rad~\cite{hartmann2020anchoring} and Braun et al.~\cite{braun2025anchoring} model deliberation as preference updating and show that outcomes can disproportionately resemble initial speakers' preferences, exhibiting anchoring as a structural bias that arises even among rational participants.

\sbpara{Nash bargaining.}
Bargaining is a fundamental concept in cooperative game theory.
Well-known solutions include the Nash bargaining solution~\cite{nash1953two}, the Kalai--Smorodinsky solution~\cite{kailai1975other}, and the egalitarian and utilitarian solutions~\cite{mas1995microeconomic}.
Nash bargaining has been widely adopted in theoretical analysis, as it is the only solution satisfying all of the following properties: (1) invariance to the origins and units of the utilities, (2) Pareto optimality, (3) symmetry, and (4) independence of irrelevant alternatives~\cite{nash1953two}.

\sbpara{Sequential deliberation and distortion.}
Fain et al.~\cite{fain2017sequential} introduced sequential deliberation as a simple and decentralized mechanism for social choice on median graphs.
Related work studies large-scale deliberative decision making and metric deliberation mechanisms~\cite{goelLargeScaleDeliberativeDecisionMaking2016,goel2025metric}.
Our work keeps the pairwise bargaining structure of sequential deliberation but adds anchoring as a behavioral feature of the current outcome.
The notion of \emph{distortion} was introduced by Procaccia and Rosenschein~\cite{procaccia2006distortion} in the voting setting to quantify the loss from choosing based on limited preference information rather than full utilities.
Metric distortion asks for similar guarantees when voters and alternatives lie in an unknown metric space~\cite{anshelevich2018approximating}.
A long line of work has improved the best achievable distortion bounds for deterministic and randomized voting rules~\cite{anshelevich2018approximating,munagala2019improved,gkatzelis2020resolving,kizilkaya2022plurality,kizilkaya2023generalized,charikar2024breaking}.

Conceptually close to our framework is the model of reality-aware social choice~\cite{shapiro2018incorporating,shapiro2026constitutional} proposed by Shapiro and Talmon.
Both models study how participants' preferences can be aggregated to update a status quo, but they use different aggregation methods. In particular, our model does not require the system to know all participants' explicit positions in the space, whereas Shapiro and Talmon's framework requires the disclosure of each participant's ideal element and the proposal's explicit position to obtain a generalized median.

\subsection{Notation}
Let $G$ be a population cumulative distribution on $[0,1]$, and let $U,V \stackrel{\mathrm{i.i.d.}}{\sim} G$.
Fix an anchoring parameter $\lambda \in [0,1)$. Conditioned on the previous outcome $O_{t-1}$, we define the anchored stances
$\Utilde \coloneqq (1-\lambda)U+\lambda O_{t-1}$ and $\Vtilde \coloneqq (1-\lambda)V+\lambda O_{t-1}$.
When we specify the previous outcome $O_{t-1}=x$, we write the corresponding updated stances as $\Utilde_x$ and $\Vtilde_x$.
Let $K_{\lambda,G}$ denote the one-step transition kernel
$K_{\lambda,G}(x,A)\coloneqq \Pr\left((1-\lambda)\med(U,V,x)+\lambda x\in A\right)$
for every Borel set $A\subseteq[0,1]$.
When there is no ambiguity, we write $K$ for $K_{\lambda,G}$ and $K^t$ for the $t$-step transition kernel.
We use $\delta_x K$, equivalently $K(x,\cdot)$, for the one-step distribution starting from $x$.

Because of space constraints, we omit the definitions and properties of the $1$-Wasserstein distance, those of the total-variation distance, and the mixing-time definitions; we refer to Appendix~\ref{app:omitted-definitions}.
We omit the one-step transition law and refer to Appendix~\ref{app:one-step-transition-law}.
All the omitted proofs are presented in the appendix.

\section{Contraction analysis}
\label{sec:transition-rule-and-mixing-time}
Our main result in this section gives the $1$-Wasserstein mixing-time bounds in Theorem~\ref{thm:wasserstein-mixing-time}; we first state the theorem and then present the coupling argument that proves it.
We put omitted proofs in Appendix~\ref{sec:omitted-proof-transition-rule-and-mixing-time}, where we also present self-contained results for the uniform distribution, including the one-step minorization argument.

\begin{theorem}
    \label{thm:wasserstein-mixing-time}
    Let $G$ be any population distribution on $[0,1]$, and let $0\le \lambda<1$.
    The chain defined by Equation~\eqref{eq:bargain-step-outcome} has a unique stationary distribution $\pi_{\lambda,G}$, and for every $0<\varepsilon<1$,
    $
    t_{\mathrm{mix}}^{(W_1)}(\varepsilon)
    \le
    \left\lceil
    \frac{\log(1/\varepsilon)}{\log(2/(1+\lambda))}
    \right\rceil.
    $
    If $0<\lambda<1$, then for every $0<\varepsilon<\frac{1}{2}$,
    $
    \left\lceil
    \frac{\log(1/(2\varepsilon))}
    {\log(1/\lambda)}
    \right\rceil
    \le
    t_{\mathrm{mix}}^{(W_1)}(\varepsilon).
    $
\end{theorem}

    Next, we show that there exists a population distribution $G$ such that the $W_1$-mixing-time upper bound is tight up to a factor of $2$.
    The construction is a symmetric two-cluster population at the endpoints.

    \begin{lemma}
    \label{lem:wasserstein-rate-tight}
    For every $0\le\lambda<1$, there exists a population distribution $G$ such that, for every $t\ge0$,
    $
        \sup_{x\in[0,1]}W_1(\delta_xK^t,\pi_{\lambda,G})
        \ge
        \frac{1}{2}\left(\frac{1+\lambda}{2}\right)^t.
    $
    Consequently,
    $t_{\mathrm{mix}}^{(W_1)}(\varepsilon) \ge \left\lceil \frac{\log(1/(2\varepsilon))}{\log(2/(1+\lambda))} \right\rceil$ for every $0<\varepsilon<\frac{1}{2}$.
    \end{lemma}

\begin{figure}[t]
    \centering
    \begin{tikzpicture}[
        x=0.82cm,
        y=0.72cm,
        >=Stealth,
        base/.style={line width=0.35pt},
        chain/.style={circle, fill=black, inner sep=1.5pt},
        drawpt/.style={circle, draw=black, fill=white, inner sep=1.4pt},
        nextpt/.style={diamond, draw=black, fill=black!12, inner sep=1.5pt},
        moveX/.style={->, line width=1pt, color=red!70},
        moveY/.style={->, line width=1pt, color=blue!70},
        caselabel/.style={anchor=east, align=right, font=\small},
        pointlabel/.style={font=\scriptsize}
    ]
        \def\x{3.4}
        \def\y{6.6}

        \foreach \yy/\case/\formula in {
            0/{Case 1: same side outside}/{$Z_{t+1}=\lambda z$},
            -2.0/{Case 2: one inside}/{$Z_{t+1}=(1-\lambda)(R-X_t)+\lambda z$},
            -4.0/{Case 3: both inside}/{$Z_{t+1}=(1-\lambda)(R-L)+\lambda z$},
            -6.0/{Case 4: opposite sides}/{$Z_{t+1}=z$}
        }{
            \draw[base] (0,\yy) -- (10,\yy);
            \node[pointlabel, below] at (0,\yy) {$0$};
            \node[pointlabel, below] at (10,\yy) {$1$};
            \node[caselabel] at (-0.35,\yy) {\case\\\formula};
            \node[chain, label={[pointlabel]above:$X_t$}] at (\x,\yy) {};
            \node[chain, label={[pointlabel]above:$Y_t$}] at (\y,\yy) {};
        }

        \node[drawpt, label={[pointlabel]above:$L$}] at (1.15,0) {};
        \node[drawpt, label={[pointlabel]above:$R$}] at (2.05,0) {};
        \node[nextpt, label={[pointlabel]below:$X_{t+1}$}] at (2.60,-0.35) {};
        \node[nextpt, label={[pointlabel]below:$Y_{t+1}$}] at (3.90,-0.35) {};
        \draw[moveX] (2.05,-0.11) -- (2.60,-0.11);
        \draw[moveY] (2.05,-0.18) -- (3.90,-0.18);

        \node[drawpt, label={[pointlabel]above:$L$}] at (1.30,-2) {};
        \node[drawpt, label={[pointlabel]above:$R$}] at (4.75,-2) {};
        \node[nextpt, label={[pointlabel]below:$X_{t+1}=X_t$}] at (\x,-2.35) {};
        \node[nextpt, label={[pointlabel]below:$Y_{t+1}$}] at (5.50,-2.35) {};
        \draw[moveY] (4.75,-2.18) -- (5.50,-2.18);

        \node[drawpt, label={[pointlabel]above:$L$}] at (4.25,-4) {};
        \node[drawpt, label={[pointlabel]above:$R$}] at (5.55,-4) {};
        \node[nextpt, label={[pointlabel]below:$X_{t+1}$}] at (3.90,-4.35) {};
        \node[nextpt, label={[pointlabel]below:$Y_{t+1}$}] at (5.95,-4.35) {};
        \draw[moveX] (4.25,-4.18) -- (3.90,-4.18);
        \draw[moveY] (5.55,-4.18) -- (5.95,-4.18);

        \node[drawpt, label={[pointlabel]above:$L$}] at (1.45,-6) {};
        \node[drawpt, label={[pointlabel]above:$R$}] at (8.55,-6) {};
        \node[nextpt, label={[pointlabel]below:$X_{t+1}=X_t$}] at (\x,-6.35) {};
        \node[nextpt, label={[pointlabel]below:$Y_{t+1}=Y_t$}] at (\y,-6.35) {};
    \end{tikzpicture}
    \caption{The four cases in the coupling argument.  
    The black points $X_t$ and $Y_t$ are fixed in each row, while $U$ and $V$ are sampled from the same distribution $G$; we let $L = \min \{U, V\}$ and $R = \max \{U, V\}$, and demonstrate the four cases with respect to the relative positions of $L$ and $R$ with respect to $X_t$ and $Y_t$. 
    The diamonds show the corresponding next states $X_{t+1}$ and $Y_{t+1}$, we use the red and blue arrows to show the selected medians, and how they are updated and become $X_{t+1}$ and $Y_{t+1}$.
    The $\lambda$ is set to be $0.4$ in this figure.
    }
    \label{fig:coupling-cases}
\end{figure}
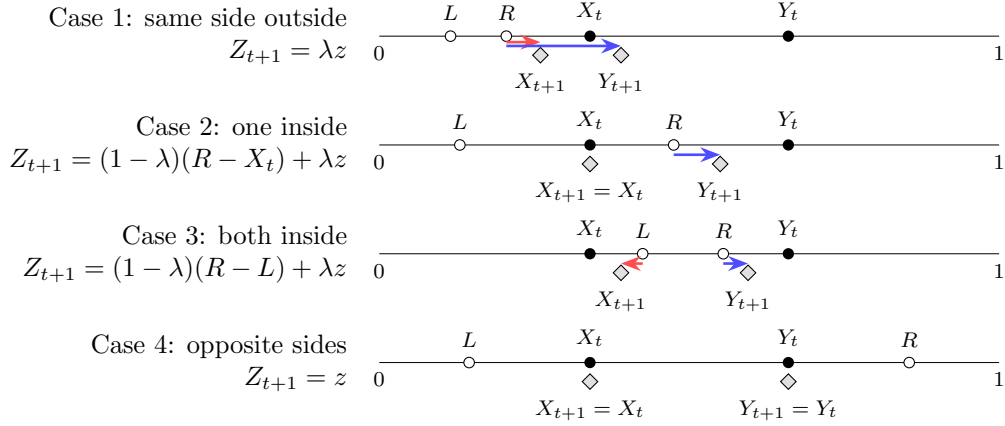

The remaining part of this section is devoted to the proof of Theorem~\ref{thm:wasserstein-mixing-time}. 
We start with the coupling argument. 
The update rule in Equation~\eqref{eq:bargain-step-outcome} can be rewritten as a linear combination of the previous outcome and the median of the two bliss points with the previous outcome:
\begin{equation}
    \label{eq:transition}
    \begin{aligned}
        O_t &= \Median(\{\Util, \Vtil, O_{t-1}\})  \\
        &= \Median\bigl(\{(1-\lambda) U + \lambda O_{t-1},\; (1-\lambda) V + \lambda O_{t-1},\; O_{t-1}\}\bigr) \\
        &= (1 - \lambda) \cdot \Median\bigl(U, V, O_{t-1}\bigr) + \lambda \cdot O_{t-1}.
    \end{aligned}
\end{equation}

Couple two copies of the Markov chain, $(X_t,Y_t)$, initialized with $X_0\le Y_0$. The extremal initialization $X_0=0$ and $Y_0=1$ gives the largest possible initial gap. At each step, update both chains with the same pair of bliss points $U_t,V_t\stackrel{\mathrm{i.i.d.}}{\sim}G$; the order of $X_t$ and $Y_t$ is preserved.
At one step of the coupling, write $L=\min\{U,V\}$ and $R=\max\{U,V\}$. Then

\begin{equation}
    \label{eq:general-contractionrule}
    \begin{aligned}
Y_{t+1} - X_{t+1} &= (1 - \lambda) \cdot (\Median\bigl(U, V, Y_{t}\bigr) - \Median\bigl(U, V, X_{t}\bigr) ) + \lambda \cdot (Y_{t} - X_{t}) \\
        & = (1-\lambda) \abs{[L, R] \cap [X_{t}, Y_{t}]} + \lambda (Y_{t} - X_{t}),
    \end{aligned}
\end{equation}
where $\abs{\cdot}$ denotes interval length. Figure~\ref{fig:coupling-cases} illustrates the four cases behind this identity; the uniform case is treated in more detail in Appendix~\ref{sec:uniform-distribution}.
The distance between the coupled chains never increases, because $\abs{[L,R]\cap[X_t,Y_t]}\le Y_t-X_t$. This also preserves the order of the coupled chains.
These geometric properties make the coupling well suited to a $1$-Wasserstein convergence analysis.

The crucial step is to express the overlap length $\abs{[L, R] \cap [X_{t}, Y_{t}]}$ as an integral:
\begin{equation}
    \label{eq:integral-range}
    \abs{[L, R] \cap [X_{t}, Y_{t}]} = \int_{X_{t}}^{Y_{t}} \ind_{L \leq r \leq R} \, dr.
\end{equation}
Taking conditional expectations in Equation~\eqref{eq:general-contractionrule} therefore reduces the computation of the overlap length to the integral of $\Pr(L\le r\le R)$ over $[X_t,Y_t]$.

\begin{lemma}
\label{lem:nonuniform-coupling}
Let $G$ be any population distribution on $[0,1]$.
Let $(X_t,Y_t)$ be the coupling obtained by using the same pair $U_t,V_t\stackrel{\mathrm{i.i.d.}}{\sim}G$ in both updates.
Let $Z_t\coloneqq \abs{X_t-Y_t}$.
If $X_0\le Y_0$, then $X_t\le Y_t$ for every $t$, and
\[
\Exp[Z_{t+1} \mid X_t, Y_t]
=
\lambda Z_t
+
(1-\lambda)\int_{X_t}^{Y_t} 2G(r)\bigl(1-G(r)\bigr)\,dr.
\]
Consequently, $\Exp[Z_{t+1}\mid X_t,Y_t]\le \frac{1+\lambda}{2}Z_t$.
\end{lemma}

The bound $\Exp[Z_{t+1} \mid X_t, Y_t] \leq \frac{1+\lambda}{2}Z_t$ follows immediately from $2G(r)(1-G(r))\le \frac{1}{2}$ for all $r\in[0,1]$. Taking expectations gives
$
\Exp[Z_{t+1}]
\le
\lambda \Exp[Z_t]+\frac{1-\lambda}{2}\Exp[Z_t]
=
\frac{1+\lambda}{2}\Exp[Z_t].
$
This coupling controls the expected distance between two copies of the chain, and hence gives a $1$-Wasserstein bound. It does not directly imply a total-variation bound, because the coupled chains become closer but their distances never become zero purely through this specific coupling analysis.

\begin{corollary}
\label{cor:wasserstein-mixing}
For every $0\le \lambda<1$ and every pair of probability measures $\mu,\nu$ on $[0,1]$, $W_1(\mu K,\nu K) \le \frac{1+\lambda}{2}\,W_1(\mu,\nu).$
The Markov chain admits a unique stationary distribution $\pi_{\lambda,G}$.
\end{corollary}

Iterating this contraction gives the upper bound in Theorem~\ref{thm:wasserstein-mixing-time}. The lower bound uses the same monotone coupling from the extremal initial states $0$ and $1$: the gap satisfies $Y_t-X_t\ge \lambda^t$, so the Kantorovich--Rubinstein duality theorem (Theorem~\ref{thm:kantorovich-rubinstein}) implies
$
W_1(\delta_0K^t,\delta_1K^t)
\ge
\abs{\Exp[Y_t]-\Exp[X_t]}
\ge
\lambda^t.
$
By applying the triangle inequality, we get the lower bound in Theorem~\ref{thm:wasserstein-mixing-time}.

\section{Stationary distortion bound}
\label{sec:worst-case}
We now compare the stationary distortions of the anchored processes under different anchoring parameters. 
For any population distribution, the stationary social cost weakly decreases as the anchoring parameter $\lambda$ increases.
Nevertheless, the same worst-case distortion bound $\frac{1+\sqrt{2}}{2}$ remains tight for every fixed $\lambda<1$.

\begin{theorem}
    \label{thm:anchored-worst-case}
    Let $G$ be any population distribution on $[0,1]$. For every $0\le \lambda_2\le \lambda_1<1$, 
let $X_{i,\infty}\sim\pi_{\lambda_i,G}$ be a stationary outcome with anchoring parameter $\lambda_i$. Then
$\mathbb{E}\!\left[\mathsf{SC}(X_{1, \infty})\right] \le \mathbb{E}\!\left[\mathsf{SC}(X_{2, \infty})\right]$.
\end{theorem}

The next lemma restates the unanchored stationary distortion bound on the line. It is consistent with Theorem~1 of Fain et al.~\cite{fain2017sequential}.

\begin{lemma}
\label{lem:unanchored-line-benchmark}
Let $X_{0, \infty}\sim\pi_{0,G}$ be the stationary outcome of the unanchored process. Then
\[
\Exp[\SCost(X_{0, \infty})]
\le
\frac{1 + \sqrt{2}}{2} \min_a \SCost(a).
\]
\end{lemma}

This bound is tight in the worst case. In particular, for every fixed anchoring parameter, there are instances whose distortion approaches $\frac{1+\sqrt{2}}{2}$.
The construction uses a two-point population with mass $p$ at $0$ and mass $q=1-p$ at $1$; when $p = \frac{\sqrt{2}}{2}$ and $q = 1 - \frac{\sqrt{2}}{2}$, the distortion is exactly $\frac{1+\sqrt{2}}{2}$ for every $\lambda<1$.

\begin{lemma}
\label{lemma:worst-case-tightness}
There exists a population distribution $G$ such that for every $0\le\lambda<1$, the distortion is exactly
$\frac{1+\sqrt{2}}{2}$. 
\end{lemma}

We outline the proof of Lemma~\ref{lem:unanchored-line-benchmark} and devote the remainder of the section to the proof of Theorem~\ref{thm:anchored-worst-case}.
We first apply Lemma~\ref{lem:social-cost-alternative} that rewrites social cost as an integral over thresholds. It uses the identity
$|O-X|=\int_0^1\bigl(\ind_{X\le r<O}+\ind_{O\le r<X}\bigr)\,dr$
and the fact that $G(r)\le 1-G(r)$ to the left of a population median, while $G(r)\ge 1-G(r)$ to its right. This also identifies population medians as social-cost minimizers.
Next, we observe that Lemma~\ref{lem:fain-general-law} gives the stationary CDF of the unanchored chain. If $F_t(r)=\Pr(X_t\le r)$, stationarity and the unanchored transition imply
$F_t(r)=G(r)^2+2G(r)(1-G(r))F_{t-1}(r)$.
Thus, Lemma~\ref{lem:unanchored-line-benchmark} follows by plugging the stationary CDF from Lemma~\ref{lem:fain-general-law} into the threshold representation of social cost in Lemma~\ref{lem:social-cost-alternative} and using
$
\max_{0\le p\le 1/2}\frac{1-p}{p^2+(1-p)^2}\le \frac{1+\sqrt2}{2}.
$
This leads to the same upper bound on the distortion as Fain et al.~\cite{fain2017sequential}.
The role of a threshold $r$ here is analogous to a coordinate cut in the median-graph argument of Fain et al.; the proof bounds the contribution of each threshold cut and then integrates over $r$.

\begin{lemma}
    \label{lem:social-cost-alternative}
Let $G$ be any population distribution on $[0,1]$, and let $O$ be any random outcome.
Then
\begin{equation}
    \label{eq:social-cost-alternative}
\Exp[\SCost(O)]
=
\int_{0}^{1}  \big( G(r) \Pr(O > r) + (1 - G(r)) \Pr(O \leq r) \big) d\,r.
\end{equation}
Moreover, for any deterministic $a\in[0,1]$,
$\SCost(a)=\int_0^a G(r)\,dr+\int_a^1(1-G(r))\,dr.$
If $m_G$ is any population median, then
$
\min_{a} \SCost(a) = \SCost(m_{G}) = \int_{0}^{1} \min \{G(r), 1 - G(r)\} d\,r.
$
\end{lemma}

\begin{lemma}
\label{lem:fain-general-law}
Let $G$ be any cumulative distribution function on $[0,1]$, and $U_t,V_t\stackrel{\mathrm{i.i.d.}}{\sim}G$.
Consider the unanchored chain when $\lambda = 0$,
$X_{t+1}=\Median\{X_t,U_t,V_t\}$.
The stationary CDF is $\Pi_G(z)=\frac{G(z)^2}{G(z)^2+\bigl(1-G(z)\bigr)^2},$ for every $z \in [0, 1]$.
\end{lemma}

\sbpara{Proof of Theorem~\ref{thm:anchored-worst-case}.}
\begin{proof}
Let $X_{1, \infty} \sim \pi_{\lambda_1, G}$ and $X_{2, \infty} \sim \pi_{\lambda_2, G}$, and let $U,V\stackrel{\mathrm{i.i.d.}}{\sim}G$.
We define the function
$$\varphi_n^{\lambda}(x) = (1 - \lambda)\mathbb E[\sum_{i = 0}^{n-1} \mathsf{SC}(Y_i) \mid Y_0 = x],$$
where $\{Y_i\}_{i = 0}^{\infty}$ is the outcome sequence with an anchoring parameter $\lambda$.
Let $M_1 = (1 - \lambda_2) \cdot \text{Med}(X_{1, \infty}, U, V) + \lambda_2 X_{1, \infty}$.
Thus, we can write $\Exp[\mathsf{SC}(X_{1, \infty})] - \Exp[\mathsf{SC}(X_{2, \infty})]$ as a telescoping identity, i.e., 
$\Exp[\mathsf{SC}(X_{1, \infty})] - \Exp[\mathsf{SC}(X_{2,\infty})] = \frac{1}{1 - \lambda_2}\lim_{n \rightarrow \infty} \{\Exp[\varphi_{n}^{\lambda_{2}}(X_{1, \infty})] - \Exp[\varphi_{n}^{\lambda_2}(M_1)] \}.$
This follows from the telescoping property and convergence to the stationary distribution of the Markov chain as $n \to \infty$. In particular,
\begin{equation}
    \begin{aligned}
    \lim_{n \to \infty} \{\Exp[\varphi^{\lambda_2}_n(X_{1, \infty})] - \Exp[\varphi^{\lambda_2}_n(M_1)] \} &\stackrel{(a)}{=} \lim_{n \to \infty} \{(1 - \lambda_2) \sum_{i = 0}^{n-1} \Exp[\Exp[\SCost(Y_i) \mid Y_0 = X_{1, \infty}]] \\
    &\qquad- (1 - \lambda_2) \sum_{i = 1}^{n} \Exp[\Exp[\SCost(Y_i) \mid Y_0 = X_{1, \infty}]] \} \\
    &\stackrel{(b)}{=} (1 - \lambda_{2}) ( \Exp[\SCost(X_{1, \infty})] - \Exp[\SCost(X_{2, \infty})] ),
    \end{aligned}
\end{equation}
where $(a)$ holds by the definition of $M_1$ and $(b)$ holds as $\lim_{n \to \infty} \Exp[\SCost(Y_n) \mid Y_0 = X_{1, \infty}] = \Exp[\SCost(X_{2, \infty})]$ for any setting of $Y_0$.

Next, we exhibit the convexity of $\varphi_{n}^{\lambda}(x)$ and that it is bounded. 

\begin{lemma}[convexity of $\varphi_{n}^{\lambda}(x)$]
For every $n\ge0$ and $\lambda \in [0, 1)$, the function $\varphi_{n}^{\lambda}$ is convex and $1$-Lipschitz on $[0,1]$.
\end{lemma}

\begin{proof}
    We use induction on $n$. The base function with $n = 1$ is  
    $\varphi_1^{\lambda}(x) = (1-\lambda) \Exp[\SCost(Y_0)\mid Y_0=x]
    = (1 - \lambda) \SCost(x).$
    Since $\SCost(x)$ is convex and $1$-Lipschitz and $0< 1- \lambda \le 1$, it follows directly
    that $\varphi^{\lambda}_1$ is convex and $(1-\lambda)$-Lipschitz, and therefore $1$-Lipschitz.
    
    We now prove the induction step. Suppose that, for some $n\ge 1$,
    $\varphi^{\lambda}_n$ is convex and $1$-Lipschitz. Conditioning on the sampled pair $U$ and $V$
    gives
    $$\varphi^{\lambda}_{n+1}(x)
    =(1 - \lambda) \SCost(x) +
    \Exp\left[ \varphi^{\lambda}_n\!\left(\lambda x+ (1 - \lambda) \med(x,U,V)\right)
    \right].
    $$

    Since $U \sim G$ and $V \sim G$, it holds by the definition of social cost that 
    $\SCost(x) = \Exp[\abs{x - U}]$ and that $\SCost(x) = \Exp[\abs{x - V}]$; 
    putting these together, we get
    $(1 - \lambda) \SCost(x)= \Exp\!\left[
    \frac{1-\lambda}{2}\bigl(|x-U|+|x-V|\bigr)
    \right].$
    
    We condition on a fixed pair $U,V$. Let
    $L=\min\{U,V\}$ and $R=\max\{U,V\}$. 
    For this pair, we consider
    
    $$
    B_{L,R}(x)
    =
    \frac{1 - \lambda}{2}\bigl(|x-L|+|x-R|\bigr)
    +
    \varphi^{\lambda}_n\!\left(\lambda x+(1 - \lambda) \med(x,L,R)\right).
    $$

    If $x<L$, then $\med(x,L,R)=L$, and the slope of $B_{L,R}$ has the form
    $-(1 - \lambda)+\lambda (\varphi^{\lambda}_n)'.$
    
    If $L<x<R$, then $\med(x,L,R)=x$, and the slope has the form
    $(\varphi^{\lambda}_n)'.$
    
    If $x>R$, then $\med(x,L,R)=R$, and the slope has the form
    $(1 - \lambda)+\lambda (\varphi^{\lambda}_n)'.$
    
    Since $\varphi^{\lambda}_n$ is convex and $1$-Lipschitz, its slopes within these three regimes are non-decreasing and lie in $[-1,1]$. Therefore, all three slopes lie in $[-1,1]$.
    
    It remains to check that the slope does not decrease when $x$ crosses $L$ or $R$. At $L$, the increase from the left slope to the right slope is not smaller than 
    $(1 - \lambda)+ (1 - \lambda) (\varphi^{\lambda}_{n, -} (L))' \geq 0$;
    at $R$, the increase is not smaller than 
    $(1 - \lambda) - (1-\lambda) (\varphi^{\lambda}_{n, +} (R))' \geq 0$. 
    Both inequalities follow because $\varphi^{\lambda}_n$ is convex and its slopes lie in $[-1,1]$. 
    Here, $(\varphi^{\lambda}_{n, -})'$ means taking the derivative from the left side and $(\varphi^{\lambda}_{n, +})'$ means taking the derivative from the right side. 
    
    Thus $B_{L,R}$ is convex and $1$-Lipschitz for every fixed pair $L,R$.
    Taking expectations over $U,V$, we can show that
    $\varphi^{\lambda}_{n+1}$ is convex and $1$-Lipschitz, and we are done. 
\end{proof}

Using $U,V$ again, we next define $\widetilde{X}_1 = (1 - \lambda_1) \cdot \text{Med}(X_{1, \infty}, U, V) + \lambda_1 \cdot X_{1, \infty}$.
By stationarity, $\widetilde{X}_1 \sim \pi_{\lambda_1, G}$;
hence $\Exp[\varphi_{n}^{\lambda_{2}}(X_{1, \infty})] = \Exp[\varphi_{n}^{\lambda_{2}}(\widetilde{X}_{1})]$.
Relating $\widetilde{X}_1$, $X_{1, \infty}$, and $M_1$, we have
$$
\widetilde{X}_1 = \left(1 - \frac{1 - \lambda_1}{1 - \lambda_2}\right) X_{1, \infty} + \frac{1 - \lambda_1}{1 - \lambda_2} M_1,
$$
which implies $\Exp[\varphi_{n}^{\lambda_{2}}(X_{1, \infty})] = \Exp[\varphi_{n}^{\lambda_{2}}(\left(1 - \frac{1 - \lambda_1}{1 - \lambda_2}\right) X_{1, \infty} + \frac{1 - \lambda_1}{1 - \lambda_2} M_1)]$.
By the convexity of $\varphi_n^{\lambda_2}$ and the linearity of expectation, we obtain
$$\Exp[\varphi_{n}^{\lambda_{2}}(X_{1, \infty})] - \Exp[\varphi_{n}^{\lambda_2}(M_1)] \leq 0.$$
Combining this inequality with the formula above completes the proof.
\end{proof}

\section{The outcome distribution}
\label{sec:general-distortion}
We now study how the stationary outcome distribution changes with the anchoring parameter. 
We show that, when $\lambda$ approaches $1$, the stationary distribution concentrates around a fixed point at which the expected one-step movement is zero. We call this point the \emph{deliberative fixed point}. 
In particular, for symmetric population distributions, this fixed point coincides with the population median; in the uniform case this yields an explicit stationary distortion bound.

\sbpara{When $\lambda$ approaches $1$.}
We show that, as $\lambda$ approaches $1$, the stationary distribution concentrates around a fixed point. For a current disagreement alternative $x$, define the expected unanchored movement
\[
b_G(x)\coloneqq \Exp[\med(x,U,V)]-x,
\]
where $U,V\stackrel{\mathrm{i.i.d.}}{\sim}G$. The deliberative fixed point $\theta_G$ is the unique point satisfying $b_G(\theta_G)=0$. Starting from $\theta_G$, the median step has no expected movement to the left or to the right.

Uniqueness follows from two facts: $b_G$ is continuous and strictly decreasing, while $b_G(0)=\Exp[\med(0,U,V)]\ge0$ and $b_G(1)=\Exp[\med(1,U,V)]-1\le0$.

\begin{lemma}
\label{lem:deliberative-fixed-point}
Let $b_G(x) \coloneqq \Exp[\med(x, U, V)] - x$, where $U, V \stackrel{\mathrm{i.i.d.}}{\sim} G$. Then there exists a unique point $\theta_G\in[0,1]$ such that $b_G(\theta_G) = 0$.
\end{lemma}

Equivalently,
\[
b_G(x)=\int_x^1 (1-G(r))^2\,d r-\int_0^x G(r)^2\,d r.
\]
This follows by writing
$\med(x,U,V)-x=\int_x^1\ind_{\med(x,U,V)\ge r}\,dr-\int_0^x\ind_{\med(x,U,V)\le r}\,dr$
and then taking expectations; for a CDF with atoms, the choice of strict or weak threshold inequalities affects only countably many values of $r$ and therefore does not change the integrals.
Thus, for $0\le x<y\le1$,
\begin{equation}
    \label{eq:monotone}
b_G(y)-b_G(x)
=-\int_x^y\left(G(r)^2+(1-G(r))^2\right)\,dr
\le -\frac{1}{2}(y-x).
\end{equation}
This monotonicity is a crucial component of the concentration result below.
The next theorem shows that concentration around the deliberative fixed point holds for every population distribution.

\begin{theorem}
    \label{thm:general-stationary-concentration}
    Let $0\le\lambda<1$ and $O\sim\pi_{\lambda,G}$. Then
    $\Exp[(O-\theta_G)^2] \le 1-\lambda$.
\end{theorem}

The proof expands
$(O_{t+1}-\theta_G)^2$ as
$((O_t-\theta_G)+(1-\lambda)(M_t-O_t))^2$
using Equation~\eqref{eq:transition}, where $M_t=\med(O_t,U_t,V_t)$. The cross term is controlled by the drift identity $\Exp[M_t-O_t\mid O_t]=b_G(O_t)$. Since $b_G(\theta_G)=0$, Equation~\eqref{eq:monotone} implies
$(x-\theta_G)b_G(x)\le -\frac{1}{2}(x-\theta_G)^2$ for every $x\in[0,1]$.

\vspace{1em}
The stationary distribution concentrates around the deliberative fixed point as $\lambda$ increases; however, this does not necessarily translate to smaller distortion, because the deliberative fixed point need not be the population median.
This concentration is most useful for reducing distortion when the two points coincide.
We therefore turn to symmetric population distributions.
We say that $G$ is centered if
$G\!\left(\frac{1}{2}+s\right)=1-G\!\left(\frac{1}{2}-s\right)$ for $0\le s\le\frac{1}{2}$.

\begin{lemma}
\label{lem:centered-fixed-point}
If $G$ is centered, then the deliberative fixed point is $\theta_G=\frac{1}{2}$.
\end{lemma}

Thus, for centered populations, anchoring concentrates the stationary distribution around the social optimum. In this case, anchoring reduces the stationary second moment around the median, which gives sharper distortion bounds.

\subsection{Uniform distribution}
\label{sec:distortion-uniform}
In the uniform case, define
$\Qcal_t \coloneqq \abs{O_t-\frac{1}{2}}$, the absolute distance from the current outcome to the population median. We show that the stationary distortion has an upper bound of $1+\frac{1-\lambda}{6(1+\lambda)}$ and a lower bound of $1+\frac{1-\lambda}{9+7\lambda}$.

\begin{theorem}
    \label{thm:distortion-uniform}
    When $G$ is uniform and $0\le\lambda<1$,
    if $O_\infty\sim\pi_{\lambda,G}$, then
    \[
    1 + \frac{1-\lambda}{9+7\lambda}
    \leq
    \frac{\Exp[\SCost(O_\infty)]}{\SCost(1/2)}
    \le
    1 + \frac{1-\lambda}{6(1+\lambda)};
    \]
    when $\lambda=0$, the distortion is exactly $\pi-2$.
\end{theorem}

The proof reduces distortion to a second-moment bound. Lemma~\ref{lem:social-cost-uniform} expresses uniform distortion exactly in terms of $\Exp[\Qcal_t^2]$. Lemma~\ref{lem:centered-second-moment} then computes the one-step conditional second moment, and Corollary~\ref{cor:upperQ} iterates the resulting upper bound.

\begin{lemma}
    \label{lem:social-cost-uniform}
    When $G$ is uniform, the distortion after $t$ steps is
    $\frac{\Exp[\SCost(O_t)]}{\SCost(1/2)} = 1+4\Exp[\Qcal_t^2].$
\end{lemma}

For the uniform distribution, we compute $\Exp[\Qcal_{t+1}^2\mid O_t=x]$ by applying the transition density in
Equation~\eqref{eq:density}, the uniform special case of Proposition~\ref{prop:general-transition-density}. This gives
    \begin{equation}
    \begin{aligned}
    \Exp[\Qcal_{t+1}^2\mid O_t=x]
    &=
    2x(1-x)\left(x-\frac{1}{2}\right)^2 \\
    &\quad+
    \frac{2}{(1-\lambda)^2}\int_{\lambda x}^{x}\left(z-\frac{1}{2}\right)^2(z-\lambda x)\,dz \\
    &\quad+
    \frac{2}{(1-\lambda)^2}\int_{x}^{1-\lambda+\lambda x}\left(z-\frac{1}{2}\right)^2(1-\lambda+\lambda x-z)\,dz.
    \end{aligned}
    \end{equation}

Evaluating the integrals and collecting terms yields the following identity.
\begin{lemma}
\label{lem:centered-second-moment}
Let $\Qcal_t\coloneqq \abs{O_t-\frac{1}{2}}$. Then
\begin{equation}
\Exp[\Qcal_{t+1}^2 \mid \Qcal_t] = \frac{(1-\lambda)^2}{48} + \frac{1+\lambda^2}{2}\Qcal_t^2 -
\frac{(1-\lambda)(3+\lambda)}{3}\Qcal_t^4.
\end{equation}
\end{lemma}

Dropping the negative fourth-moment term in Lemma~\ref{lem:centered-second-moment} gives a linear recursion for the second moment. Iterating that recursion gives the following upper bound.
The lower bound follows by observing that $\Qcal_t^4 \leq \frac{1}{4}\Qcal_t^2$.

\begin{corollary}
    \label{cor:upperQ}
    For any $t \geq 1$,
    \[
    \frac{(1-\lambda)^2}{48}
    +\frac{7\lambda^2+2\lambda+3}{12}\Exp[\Qcal_t^2]
    \leq \Exp[\Qcal_{t+1}^2]
    \leq
    \frac{(1 -\lambda)^2}{48}
    +\frac{1 + \lambda^2}{2}\Exp[\Qcal_{t}^2].
    \]
    As $t\rightarrow \infty$, $\frac{1-\lambda}{4(9+7\lambda)} \leq \Exp[\Qcal_\infty^2] \leq \frac{1 - \lambda}{24(1+\lambda)}$.
\end{corollary}

The limiting upper bound $\frac{1-\lambda}{24(1+\lambda)}$ decreases monotonically in $\lambda$. Together with Lemma~\ref{lem:social-cost-uniform}, this proves Theorem~\ref{thm:distortion-uniform}.

\subsection{Symmetric distributions}
\label{sec:distortion-symmetric}
We now extend the second-moment analysis to population distributions that are symmetric around $1/2$.
Throughout this subsection, write again $\Qcal_t=\abs{O_t-\frac{1}{2}}$.
Assume further that $G$ is absolutely continuous on $(0,1)$ with density $g$ satisfying
$g\!\left(\frac{1}{2}+s\right)=g\!\left(\frac{1}{2}-s\right)$ for $s\in[0,1/2)$.
Under these assumptions, the main symmetric-distribution bound is the following.

\begin{theorem}
    \label{thm:centered-second-moment-general-g}
    Assume that $G$ is centered and satisfies the absolute-continuity and symmetric-density conditions stated above.
    Let $U\sim G$.
    For every $t\ge0$,
$\Exp[\Qcal_t^2]
\le
\frac{1}{4}\left(\frac{1+\lambda^2}{2}\right)^t+
\frac{2(1-\lambda)}{1+\lambda}
\Var(U).$
At stationarity,
$
\Exp[\Qcal_\infty^2]
\le
\frac{2(1-\lambda)}{1+\lambda} \Var(U).
$
\end{theorem}

The proof uses the following one-step analogue of Lemma~\ref{lem:centered-second-moment}.

\begin{lemma}
\label{lem:centered-second-moment-general-g}
Assume that $G$ is centered and satisfies the absolute-continuity and symmetric-density conditions stated above.
\begin{equation}
\begin{aligned}
\Exp[\Qcal_{t+1}^2] &\le (1-\lambda)^2 4\int_0^{1/2}\left(\frac{1}{2}-u\right)^2G(u)g(u)\,du + \frac{1+\lambda^2}{2} \Exp[\Qcal_t^2].
\end{aligned}
\end{equation}
\end{lemma}

The proof follows the same strategy as the proof of Lemma~\ref{lem:centered-second-moment}: write $\Exp[\Qcal_{t+1}^2\mid \Qcal_t]$ using the transition density from Proposition~\ref{prop:general-transition-density}, integrate the squared distance to the median, and then use symmetry to simplify the expressions in the computation.
The computation itself is more involved because it requires multiple applications of symmetries.

Because $G(u)\le\frac{1}{2}$ for $u\le\frac{1}{2}$ and $g$ is symmetric,
\[
4\int_0^{1/2}\left(\frac{1}{2}-u\right)^2G(u)g(u)\,du
\le
2\int_0^{1/2}\left(\frac{1}{2}-u\right)^2g(u)\,du
\le
\Var(U).
\]
Therefore
\begin{equation}
    \begin{aligned}
    \Exp[\Qcal_{t+1}^2]
    &\leq (1-\lambda)^2 \Var(U) + \frac{1+\lambda^2}{2} \Exp[\Qcal_t^2],
    \end{aligned}
\end{equation}
and Theorem~\ref{thm:centered-second-moment-general-g} follows by iterating this recursion.

\section{Simulation results}
\label{sec:simulations}
We approximate the anchored deliberation Markov chain by discretizing $[0,1]$ into $5000$ equal bins and constructing the transition matrix from the conditional CDF of the one-step update.
For each population distribution and each $\lambda\in\{0.00,0.01,\ldots,0.99\}$, we compute an approximation of the stationary distribution by iterating the transition matrix from the uniform distribution up to the $1$-Wasserstein mixing-time upper bound, using tolerance $\varepsilon_{\mathrm{mix}}=1/5000$.
The empirical mixing time is the first step at which the discretized $1$-Wasserstein distance to this stationary approximation is at most $\varepsilon_{\mathrm{mix}}$.
Further implementation details and additional Beta instances are reported in Appendix~\ref{app:more_simulations}.

The simulations illustrate the same trade-off as the theoretical results.
Figure~\ref{fig:beta-dynamics-main} gives a representative example for $\mathrm{Beta}(0.5,0.7)$.
As $\lambda$ increases, the empirical mixing becomes slower, the stationary median moves toward the deliberative fixed point, and the stationary distribution becomes more concentrated around that point.
Figure~\ref{fig:beta-distortion-simulations} shows that anchoring reduces stationary distortion: for the uniform distribution, the empirical curve lies between the theoretical upper and lower bounds, and the same decreasing pattern of stationary distortion appears across five Beta population laws.

\begin{figure}[t]
    \centering
    \begin{subfigure}[t]{0.32\textwidth}
        \centering
        \includegraphics[width=\linewidth]{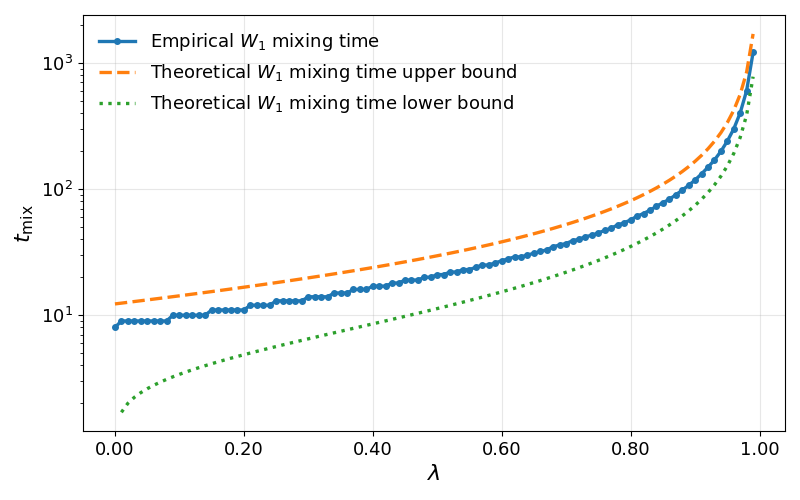}
        \caption{Mixing time}
    \end{subfigure}
    \begin{subfigure}[t]{0.32\textwidth}
        \centering
        \includegraphics[width=\linewidth]{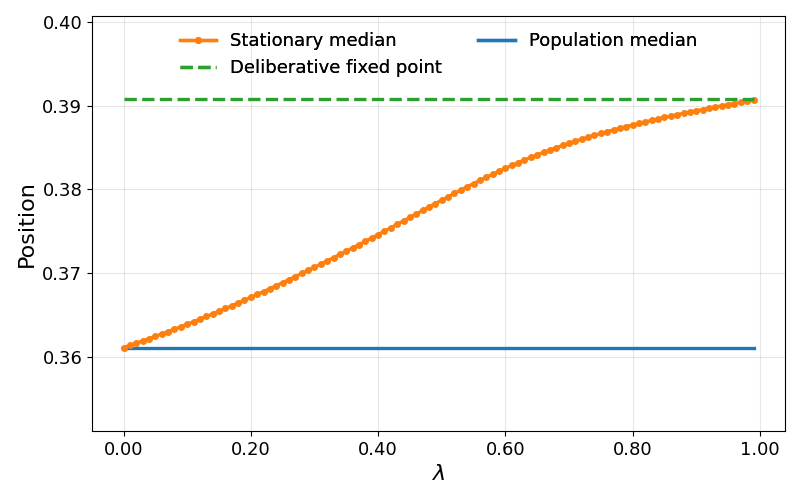}
        \caption{Medians and fixed point}
    \end{subfigure}
    \begin{subfigure}[t]{0.32\textwidth}
        \centering
        \includegraphics[width=\linewidth]{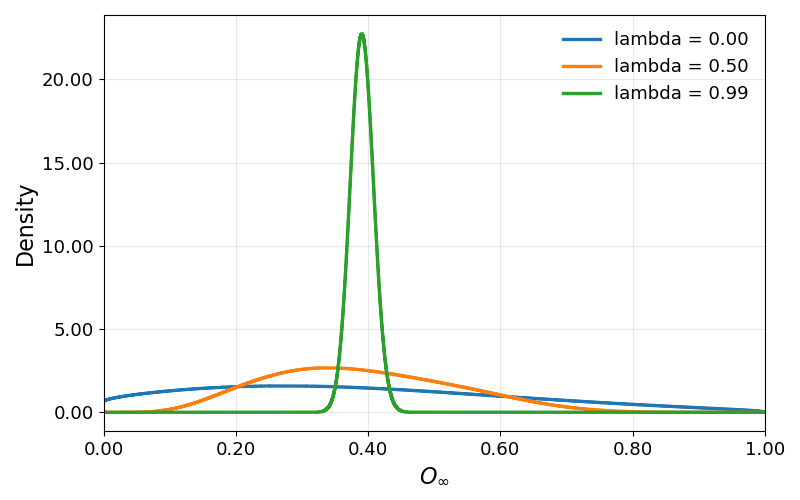}
        \caption{Stationary density}
    \end{subfigure}
    \caption{Representative simulation results for $\mathrm{Beta}(0.5,0.7)$. Empirical mixing time increases with $\lambda$ and is plotted on a logarithmic scale against the theoretical bounds. The stationary median moves toward the deliberative fixed point, and stronger anchoring concentrates the stationary distribution.}
    \label{fig:beta-dynamics-main}
\end{figure}

\begin{figure}[t]
    \centering
    \begin{subfigure}[t]{0.45\textwidth}
        \centering
        \includegraphics[width=\linewidth]{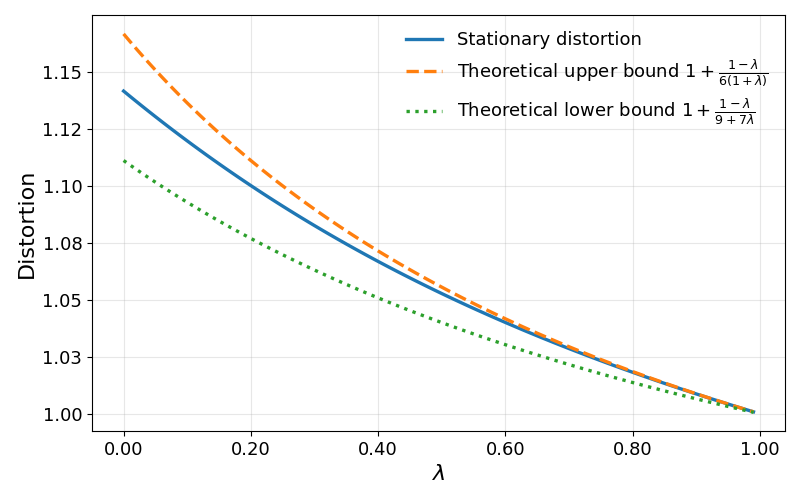}
        \caption{Uniform population}
    \end{subfigure}
    \hfill
    \begin{subfigure}[t]{0.45\textwidth}
        \centering
        \includegraphics[width=\linewidth]{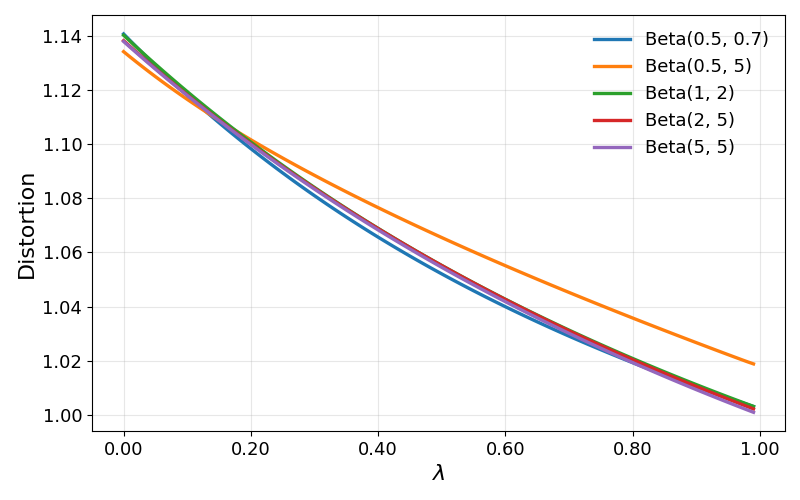}
        \caption{Selected Beta populations}
    \end{subfigure}
    \caption{Stationary distortion under anchored sequential deliberation. In the uniform case, the empirical distortion is compared with the theoretical upper and lower bounds. Across the selected Beta populations, stationary distortion decreases as anchoring becomes stronger.}
    \label{fig:beta-distortion-simulations}
\end{figure}

\section{Conclusion}
\label{sec:conclusion}
This paper studies how anchoring changes sequential deliberation.
Anchoring slows convergence 
but it can improve the long-run quality of the decision.
At stationarity, the anchored process has social cost no larger than that of the unanchored process for every fixed population distribution, and stronger anchoring concentrates outcomes around a unique deliberative fixed point.
Several extensions remain open.
First, we currently assume that $\lambda$ is the same constant for all individuals.
This assumption may not hold in practice, since anchoring can vary across people; future work could allow different individuals or groups to have different values of $\lambda$.
Second, the model could be extended from the line to median graphs, as in Fain et al.~\cite{fain2017sequential}.
The main technical difficulty is defining anchored stances on the median graph: unlike on the line, an interpolation between a bliss point and the current outcome may not be unique.
Third, future work could study the use of AI tools.
Recent work incorporates AI tools into group decision-making processes~\cite{tessler2024ai,fish2026generative,boehmer2025generative,de2026question}, and a natural question is whether such tools can be integrated into sequential deliberation to improve deliberative outcomes.

\section*{Acknowledgments}
We used ChatGPT with the GPT-5.6 Sol model to assist with simulation code and parts of the theoretical analysis. It helped generalize results from the uniform distribution to general continuous distributions and obtain the monotonicity property of the stationary distortion under the authors' supervision.
This research was funded by the Knut and Alice Wallenberg Foundation.

\bibliographystyle{alpha}
\bibliography{ref}

\clearpage

\appendix

\section{Omitted properties of the distance functions}
\label{app:omitted-definitions}

The convergence analysis uses both total variation distance and $1$-Wasserstein distance.
We use the following form of the $1$-Wasserstein distance, adapted from \cite[Definition 6.1]{villani2009optimal}.

\begin{definition}
    \label{def:wasserstein}
    Let $\mathcal{P}([0, 1])$ be the set of probability measures on the interval $[0, 1]$.
    The $1$-Wasserstein distance between $\mu$, $\nu \in \mathcal{P}([0, 1])$ is
    \begin{equation}
        \begin{aligned}
        W_1(\mu, \nu) &\coloneqq \inf_{\pi \in C(\mu, \nu)} \int_{[0, 1] \times [0, 1]} |x - y | d \,\pi(x, y) \\
        &= \inf  \Exp[\abs{X - Y}] , \text{where $(X, Y)$ is a coupling of $\mu$ and $\nu$.}
        \end{aligned}
    \end{equation}
    Here $C(\mu, \nu)$ is the set of all joint probability measures on $[0,1] \times [0, 1]$ whose marginals are $\mu$ and $\nu$.
\end{definition}

We also use the Kantorovich--Rubinstein duality theorem, adapted from \cite[Remark 6.5]{villani2009optimal}.

\begin{theorem}[Kantorovich--Rubinstein duality]
    \label{thm:kantorovich-rubinstein}
    Let $\mu,\nu\in\mathcal{P}([0,1])$.
    Then the $1$-Wasserstein distance satisfies the dual representation
    \[
    W_1(\mu,\nu)
    =
    \sup_{f\in \mathrm{Lip}_1([0,1])}
    \left\{
    \int_{[0,1]} f(x)\,d\mu(x)
    -
    \int_{[0,1]} f(x)\,d\nu(x)
    \right\}.
    \]
    Here $\mathrm{Lip}_1([0,1])$ is the set of all Lipschitz functions on $[0,1]$ with Lipschitz constant $1$.
    \end{theorem}

Plugging either $f(x)=x$ or $f(x)=-x$ into Theorem~\ref{thm:kantorovich-rubinstein} gives the following corollary.
\begin{corollary}
    \label{cor:wasserstein-dominates-mean-difference}
    For any $\mu,\nu\in\mathcal{P}([0,1])$, if $X\sim\mu$ and $Y\sim\nu$, then
    \[
    W_1(\mu,\nu)\ge
    \abs{\Exp[X]-\Exp[Y]}.
    \]
    \end{corollary}

The total-variation distance is adapted from \cite[Proposition 4.7]{levin2026markov}.

\begin{definition}
    \label{def:total-variation}
    Let $\mathcal{P}([0, 1])$ be the set of probability measures on the interval $[0, 1]$.
    The total-variation distance between $\mu$, $\nu \in \mathcal{P}([0, 1])$ is
    \begin{equation}
        \begin{aligned}
        \|\mu - \nu\|_{\mathrm{TV}} &\coloneqq \sup_{A \subseteq [0, 1]} \abs{\mu(A) - \nu(A)} \\
        &= \inf \Pr(X \neq Y), \text{where $(X, Y)$ is a coupling of $\mu$ and $\nu$.}
        \end{aligned}
    \end{equation}
\end{definition}

For any stationary distribution $\pi_{\lambda,G}$, define the $1$-Wasserstein mixing time by
\[
t_{\mathrm{mix}}^{(W_1)}(\varepsilon)
\coloneqq
\min\Bigl\{t:\sup_{x\in[0,1]}W_1(\delta_xK^t, \pi_{\lambda,G})\le \varepsilon\Bigr\}.
\]

The total-variation mixing time is defined as follows.
\[
t_{\mathrm{mix}}^{(\mathrm{TV})}(\varepsilon)
\coloneqq
\min\Bigl\{t:\sup_{x\in[0,1]}\|\delta_xK^t - \pi_{\lambda,G}\|_{\mathrm{TV}}\le \varepsilon\Bigr\}.
\]

\section{One-step transition law}
\label{app:one-step-transition-law}
We first compute the conditional CDF of an anchored stance.
The same formula applies to $\Utilde_x$ and $\Vtilde_x$.
\begin{proposition}
\label{prop:general-stance-law}
For every $x,z \in [0,1]$, we denote the conditional CDF of $\Utilde_x$ (as well as $\Vtilde_x$) as $F_x(z)$.
\[
\begin{aligned}
F_x(z)
&\coloneqq
\Pr(\Utilde_x \le z)
=
\begin{cases}
0, & z < \lambda x,\\[2mm]
G\!\left(\dfrac{z-\lambda x}{1-\lambda}\right), & \lambda x \le z \le \lambda x+(1-\lambda),\\[3mm]
1, & z > \lambda x+(1-\lambda).
\end{cases}
\end{aligned}
\]
\end{proposition}

We next compute the conditional CDF of the updated outcome $O_t$. The median structure gives the two cases. If $z<x$, then $\med\{\Utilde_x,\Vtilde_x,x\}\le z$ only when both anchored stances are at most $z$. If $z\ge x$, then the complementary event $O_t>z$ occurs only when both anchored stances exceed $z$. The atom at $x$ is the jump of this CDF.

\begin{proposition}
\label{prop:general-transition-cdf}
For every $x,z \in [0,1]$,
\[
\Pr(O_t \le z \mid O_{t-1}=x)
=
\begin{cases}
F_x(z)^2, & z < x,\\[2mm]
1-\bigl(1-F_x(z)\bigr)^2, & z \ge x.
\end{cases}
\]
If $G$ is continuous at $x$,
\[
\Pr(O_t=x \mid O_{t-1}=x)=2G(x)\bigl(1-G(x)\bigr).
\]
\end{proposition}

\begin{proposition}
    \label{prop:general-transition-density}
    Assume that $G$ is absolutely continuous with density $g$.
    Conditioned on $O_{t-1}=x$, the transition kernel has an atom of mass
    $2G(x)(1-G(x))$ at $x$.
    Away from this atom, its density is
    \[
    k_x(z)=
    \begin{cases}
    \dfrac{2}{1-\lambda}
    F_x(z)\,
    g\!\left(\dfrac{z-\lambda x}{1-\lambda}\right),
    & \lambda x<z<x,\\[3mm]
    \dfrac{2}{1-\lambda}
    \bigl(1-F_x(z)\bigr)\,
    g\!\left(\dfrac{z-\lambda x}{1-\lambda}\right),
    & x<z<\lambda x+(1-\lambda),\\[3mm]
    0, & \text{otherwise}.
    \end{cases}
    \]
\end{proposition}

\sbpara{Proof of Proposition~\ref{prop:general-stance-law}}
\begin{proof}
    Condition on $O_{t-1}=x$. If $z<\lambda x$, then
    $(1-\lambda)U+\lambda x\le z$ is impossible, so $F_x(z)=0$.
    If $z>\lambda x+(1-\lambda)$, then the event is certain, so $F_x(z)=1$.
    For $\lambda x\le z\le \lambda x+(1-\lambda)$, the event
    $(1-\lambda)U+\lambda x\le z$ is equivalent to
    $U\le (z-\lambda x)/(1-\lambda)$.
    Hence
    \[
        F_x(z)=G\!\left(\frac{z-\lambda x}{1-\lambda}\right),
    \]
    which proves the stated formula.
\end{proof}

\sbpara{Proof of Proposition~\ref{prop:general-transition-cdf}}
\begin{proof}
    Condition on $O_{t-1}=x$ and write
    $O_t=\med(\Utilde_x,\Vtilde_x,x)$, where $\Utilde_x$ and $\Vtilde_x$ are independent with CDF $F_x$.
    If $z<x$, then $O_t\le z$ if and only if both anchored stances are at most $z$.
    Hence
    \[
    \Pr(O_t\le z\mid O_{t-1}=x)=F_x(z)^2.
    \]
    If $z\ge x$, then $O_t>z$ if and only if both anchored stances are larger than $z$.
    Therefore
    \[
    \Pr(O_t\le z\mid O_{t-1}=x)
    =
    1-\Pr(\Utilde_x>z,\Vtilde_x>z)
    =
    1-\bigl(1-F_x(z)\bigr)^2.
    \]
    These two identities prove the stated CDF.

    If $G$ is continuous at $x$, then $F_x$ is continuous at $x$ and $F_x(x)=G(x)$.
    The atom at $x$ is the jump of the CDF:
    \[
    \begin{aligned}
    \Pr(O_t=x\mid O_{t-1}=x)
    &=
    \Pr(O_t\le x\mid O_{t-1}=x)-\Pr(O_t<x\mid O_{t-1}=x)\\
    &=
    \bigl[1-\bigl(1-F_x(x)\bigr)^2\bigr]-F_x(x)^2\\
    &=
    2F_x(x)\bigl(1-F_x(x)\bigr)
    =
    2G(x)\bigl(1-G(x)\bigr).
    \end{aligned}
    \]
\end{proof}

\sbpara{Proof of Proposition~\ref{prop:general-transition-density}}
\begin{proof}
    Differentiate the two continuous branches of Proposition~\ref{prop:general-transition-cdf}.
    By the chain rule,
    \[
    F_x'(z)=\frac{1}{1-\lambda}\,
    g\!\left(\frac{z-\lambda x}{1-\lambda}\right)
    \]
    on the interior of $[\lambda x,\lambda x+(1-\lambda)]$.
    For $z<x$,
    \[
        \frac{d}{dz}F_x(z)^2=2F_x(z)F_x'(z),
    \]
    and for $z>x$,
    \[
        \frac{d}{dz}\left(1-\bigl(1-F_x(z)\bigr)^2\right)
        =
        2\bigl(1-F_x(z)\bigr)F_x'(z).
    \]
    Substituting the expression for $F_x'(z)$ gives the stated density away from the atom at $x$.
\end{proof}

\section{Omitted proofs from Section~\ref{sec:transition-rule-and-mixing-time}}
\label{sec:omitted-proof-transition-rule-and-mixing-time}
This appendix contains the convergence details omitted from Section~\ref{sec:transition-rule-and-mixing-time}.
We put the proofs of Theorem~\ref{thm:wasserstein-mixing-time} after its supporting lemmas.
The material specific to the uniform distribution, including the one-step minorization argument, is presented in Appendix~\ref{sec:uniform-distribution}.

\subsection{Proof of Lemma~\ref{lem:wasserstein-rate-tight}}
\begin{proof}
    We consider a distribution with mass $\frac{1}{2}$ at $0$ and $\frac{1}{2}$ at $1$.
    Conditional on the current state $O_t=x$, the next state is
    \[
    O_{t+1}
    =
    \begin{cases}
    \lambda x, & \text{with probability }1/4,\\
    x, & \text{with probability }1/2,\\
    \lambda x+(1-\lambda), & \text{with probability }1/4.
    \end{cases}
    \]
    Hence
    \[
        \Exp[O_{t+1}\mid O_t=x]
        =
        \frac{1+\lambda}{2}x+\frac{1-\lambda}{4}.
    \]
    If $O\sim\pi_{\lambda,G}$ is stationary, then
    \[
        \Exp[O]=\frac{1+\lambda}{2}\Exp[O]+\frac{1-\lambda}{4},
    \]
    so $\Exp[O]=\frac{1}{2}$.
    Starting from $O_0=1$, the mean recursion gives
    \[
        \Exp[O_t]-\Exp[O]
        =
        \frac{1}{2}\left(\frac{1+\lambda}{2}\right)^t.
    \]
    By Corollary~\ref{cor:wasserstein-dominates-mean-difference},
    \[
        W_1(\delta_1K^t,\pi_{\lambda,G})
        \ge
        \abs{\Exp[O_t]-\Exp[O]}
        =
        \frac{1}{2}\left(\frac{1+\lambda}{2}\right)^t.
    \]
    The stated mixing-time lower bound follows by solving
    $\frac{1}{2}((1+\lambda)/2)^t\le \varepsilon$ for $t$.
\end{proof}

\subsection{Proof of Lemma~\ref{lem:nonuniform-coupling}}
\begin{proof}
    Write $L=\min\{U,V\}$ and $R=\max\{U,V\}$.
    Assume first that $X_t\le Y_t$; the case $X_t=Y_t$ is immediate.
    By the transition rule in Equation~\eqref{eq:transition},
    \[
    \begin{aligned}
    Y_{t+1}-X_{t+1}
    &=
    (1-\lambda)\bigl(\Median(U,V,Y_t)-\Median(U,V,X_t)\bigr)
    +\lambda(Y_t-X_t)\\
    &=
    (1-\lambda)\abs{[L,R]\cap[X_t,Y_t]}+\lambda(Y_t-X_t).
    \end{aligned}
    \]
    Since the overlap length is nonnegative, $Y_{t+1}\ge X_{t+1}$.
    Since the overlap length is at most $Y_t-X_t$, we also have
    $Y_{t+1}-X_{t+1}\le Y_t-X_t$.
    Thus the order is preserved and the gap is non-increasing.

    Conditional on $X_t,Y_t$, Fubini's theorem gives
    \[
    \begin{aligned}
    \Exp\!\left[\abs{[L,R]\cap[X_t,Y_t]}\mid X_t,Y_t\right]
    &=
    \Exp\!\left[\int_{X_t}^{Y_t}\ind_{L\le r\le R}\,dr \,\middle|\, X_t,Y_t\right]\\
    &=
    \int_{X_t}^{Y_t}\Pr(L\le r\le R)\,dr.
    \end{aligned}
    \]
    For an arbitrary CDF $G$,
    \[
        \Pr(L\le r\le R)=1-(1-G(r))^2-G(r^-)^2.
    \]
    This equals $2G(r)(1-G(r))$ at every continuity point of $G$.
    A CDF has at most countably many discontinuities, so the equality holds for Lebesgue-a.e.\ $r$ and hence after integration.
    Therefore
    \[
    \Exp[Z_{t+1}\mid X_t,Y_t]
    =
    \lambda Z_t
    +
    (1-\lambda)\int_{X_t}^{Y_t}2G(r)(1-G(r))\,dr.
    \]
    The same argument applies when $X_t>Y_t$ after swapping the labels.
    Since $2G(r)(1-G(r))\le 1/2$, the final inequality follows.
\end{proof}

\subsection{Proof of Corollary~\ref{cor:wasserstein-mixing}}
\begin{proof}
    Let $c=(1+\lambda)/2<1$.
    For deterministic $x,y\in[0,1]$, update two copies from $x$ and $y$ using the same pair $U,V\sim G$.
    Lemma~\ref{lem:nonuniform-coupling}, with the labels ordered if necessary, gives
    \[
        \Exp\!\left[\abs{T(x;U,V)-T(y;U,V)}\right]\le c\abs{x-y},
    \]
    where $T(x;U,V)=\lambda x+(1-\lambda)\med(x,U,V)$.
    Now take any coupling $(X_0,Y_0)$ of $\mu$ and $\nu$ and update both coordinates with the same $U,V$.
    Then
    \[
        \Exp[\abs{X_1-Y_1}]\le c\,\Exp[\abs{X_0-Y_0}].
    \]
    Taking the infimum over all couplings of $\mu$ and $\nu$ gives
    \[
        W_1(\mu K,\nu K)\le c\,W_1(\mu,\nu).
    \]
    Since $[0,1]$ is compact, $(\mathcal P([0,1]),W_1)$ is complete.
    Thus $\mu\mapsto \mu K$ is a contraction on a complete metric space, so Banach's fixed-point theorem~\cite[Theorem 1.1]{granas2003fixed} gives a unique fixed point $\pi_{\lambda,G}$, equivalently a unique stationary distribution.
\end{proof}

\subsection{Proof of Theorem~\ref{thm:wasserstein-mixing-time}}
\begin{proof}
    \sbpara{Upper bound.}
    Let $c=(1+\lambda)/2$.
    By Corollary~\ref{cor:wasserstein-mixing}, for every $x\in[0,1]$,
    \[
        W_1(\delta_xK^t,\pi_{\lambda,G})
        =
        W_1(\delta_xK^t,\pi_{\lambda,G}K^t)
        \le
        c^tW_1(\delta_x,\pi_{\lambda,G})
        \le
        c^t,
    \]
    because the diameter of $[0,1]$ is one.
    Therefore
    \[
        t_{\mathrm{mix}}^{(W_1)}(\varepsilon)
        \le
        \left\lceil
        \frac{\log(1/\varepsilon)}{\log(2/(1+\lambda))}
        \right\rceil.
    \]

    \sbpara{Lower bound.}
    Assume $0<\lambda<1$.
    Couple two copies of the chain by using the same sampled pair of bliss points at every time step.
    Let one copy start from $X_0=0$ and the other from $Y_0=1$.
    Since the coupling preserves order,
    \[
        Y_{t+1}-X_{t+1}\ge \lambda(Y_t-X_t).
    \]
    Iterating gives $Y_t-X_t\ge \lambda^t$.
    Hence
    \[
        \Exp[Y_t]-\Exp[X_t]
        =
        \Exp[Y_t-X_t]
        \ge
        \lambda^t.
    \]
    The laws of $X_t$ and $Y_t$ are respectively $\delta_0K^t$ and $\delta_1K^t$.
    By Corollary~\ref{cor:wasserstein-dominates-mean-difference},
    \[
        W_1(\delta_0K^t,\delta_1K^t)\ge \lambda^t.
    \]
    The triangle inequality gives
    \[
    W_1(\delta_0K^t,\delta_1K^t)
    \le
    W_1(\delta_0K^t,\pi_{\lambda,G})
    +
    W_1(\pi_{\lambda,G},\delta_1K^t)
    \le
    2\sup_{x\in[0,1]}W_1(\delta_xK^t,\pi_{\lambda,G}).
    \]
    Thus any $t$ with
    $\sup_x W_1(\delta_xK^t,\pi_{\lambda,G})\le\varepsilon$
    must satisfy $\lambda^t\le 2\varepsilon$.
    Solving this inequality gives
    \[
        t_{\mathrm{mix}}^{(W_1)}(\varepsilon)
        \ge
        \left\lceil
        \frac{\log(1/(2\varepsilon))}
        {\log(1/\lambda)}
        \right\rceil
    \]
    for every $0<\varepsilon<1/2$.
\end{proof}

\subsection{Uniform distribution}
\label{sec:uniform-distribution}
We now specialize the transition law to the case where the population distribution is uniform on $[0,1]$.
Thus the two sampled bliss points are independent draws $U,V\sim\mathrm{Unif}[0,1]$.
Conditionally on $O_{t-1}=x$, the anchored stances $\Util$ and $\Vtil$ are i.i.d.\ uniform on $[\lambda x,\lambda x+(1-\lambda)]$ with density $1/(1-\lambda)$.

For any real value $z$, the conditional CDF of either anchored stance is
\begin{equation}
    \begin{aligned}
    \Pr(\Util \leq z \mid O_{t-1} = x) = \Pr(\Vtil \leq z \mid O_{t-1} = x) =
    \begin{cases}
        0, & z < \lambda x, \\
        \dfrac{z - \lambda x}{1-\lambda}, & \lambda x \leq z < \lambda x + (1-\lambda), \\
        1, & z \geq \lambda x + (1-\lambda).
    \end{cases}
    \end{aligned}
\end{equation}
Applying Proposition~\ref{prop:general-transition-cdf} gives the transition CDF
\begin{equation}
    \label{eq:cdf}
\Pr(O_t \leq z \mid O_{t-1} = x) =
\begin{cases}
  0, & z < \lambda x, \\
  \left(\dfrac{z - \lambda x}{1-\lambda}\right)^2, & \lambda x \leq z < x, \\
  1 - \left(1 - \dfrac{z - \lambda x}{1-\lambda}\right)^2, & x \leq z < \lambda x + (1-\lambda), \\
  1, & z \geq \lambda x + (1-\lambda).
\end{cases}
\end{equation}
Differentiating away from the atom at $x$ gives the density
\begin{equation}
    \label{eq:density}
k_x(z)=
\begin{cases}
\dfrac{2({z-\lambda x})}{(1-\lambda)^2}, & \lambda x<z<x,\\[2mm]
\dfrac{2(1-\lambda+\lambda x-z)}{(1-\lambda)^2}, & x<z<\lambda x+(1-\lambda),\\[2mm]
0, & \text{otherwise},
\end{cases}
\end{equation}
plus an atom of mass $2x(1-x)$ at $x$.

We record two refinements for the uniform case. The first decomposes the one-step coupling by the position of the sampled bliss points relative to the current gap. The second is a one-step minorization that gives total-variation contraction when $\lambda<1/2$.

Let $(X_t,Y_t)$ be the coupling obtained by using the same pair $U_t,V_t\stackrel{\mathrm{i.i.d.}}{\sim}\mathrm{Unif}[0,1]$ in both updates.
Let $Z_t\coloneqq \abs{X_t-Y_t}$.

\begin{lemma}
    \label{lem:coupling-concentration}
    Condition on $X_t\le Y_t$ and $Z_t=z$.
    The order of the coupled chains is preserved and $Z_t$ is non-increasing.
    Moreover, the one-step gap satisfies the following four-case decomposition:
    with probability at least $\frac{(1-z)^2}{2}$, $Z_{t+1}=\lambda z$;
    with probability $2z(1-z)$, $\Exp[Z_{t+1}\mid Z_t=z,\text{Case 2}]=\frac{1+\lambda}{2}z$;
    with probability $z^2$, $\Exp[Z_{t+1}\mid Z_t=z,\text{Case 3}]=\frac{1+2\lambda}{3}z$;
    and with probability at most $\frac{(1-z)^2}{2}$, $Z_{t+1}=z$.
\end{lemma}

\begin{corollary}
    \label{cor:coupling-concentration-expectation}
    For every $z\in[0,1]$,
\[
    \Exp[Z_{t+1}\mid Z_t=z]\le \frac{1+\lambda}{2}z.
\]
Consequently,
    \[
    \Exp[Z_t]\le \left(\frac{1+\lambda}{2}\right)^t \Exp[Z_0].
    \]
\end{corollary}

\sbpara{Proof of Lemma~\ref{lem:coupling-concentration}}

\begin{proof}
    Write $L=\min\{U,V\}$ and $R=\max\{U,V\}$.
    Condition on $X_t\le Y_t$ and $Z_t=Y_t-X_t=z$.
    Equation~\eqref{eq:transition} gives
    \begin{equation}
        \label{eq:contractionrule}
        \begin{aligned}
    Y_{t+1} - X_{t+1}
    &=
    (1-\lambda)\bigl(\Median(U,V,Y_t)-\Median(U,V,X_t)\bigr)+\lambda(Y_t-X_t)\\
    &=
    (1-\lambda)\abs{[L,R]\cap[X_t,Y_t]}+\lambda(Y_t-X_t).
        \end{aligned}
    \end{equation}
    The overlap length is between $0$ and $z$, so $0\le Y_{t+1}-X_{t+1}\le z$.
    Hence the order is preserved and the gap is non-increasing.

    \sbpara{Case 1.} Both sampled points are outside $[X_t,Y_t]$ and on the same side, i.e., $R<X_t$ or $L>Y_t$.
    Let $\alpha=X_t/(1-z)$ and $\beta=(1-Y_t)/(1-z)$ when $z<1$; then $\alpha,\beta\ge0$ and $\alpha+\beta=1$.
    The probability of this case is
    \begin{equation}
        \begin{aligned}
            \Pr(R < X_t \text{ or } L > Y_t)
            &\geq \min_{\alpha, \beta \geq 0, \alpha + \beta = 1} ((1 - z)\alpha)^2 + ((1 - z) \beta)^2 \\
            &= \min_{\alpha, \beta \geq 0, \alpha + \beta = 1} (1 - z)^2 (\alpha^2 + \beta^2)  \\
            &\geq (1 - z)^2 (\frac{\alpha + \beta}{2})^2 \cdot 2 \\
            &= (1 -z)^2 \frac{1}{2}.
        \end{aligned}
    \end{equation}
    When $z=1$, the same lower bound is $0$ and is immediate.
    In this case the overlap length is $0$, so $Z_{t+1}=\lambda z$.

    \sbpara{Case 2.} Exactly one sampled point lies in $[X_t,Y_t]$.
    This has probability $2z(1-z)$.
    Let $\gamma=\abs{[L,R]\cap[X_t,Y_t]}$.
    Then $Z_{t+1}=(1-\lambda)\gamma+\lambda z$, and the contribution to the expectation is
    \begin{equation}
        \begin{aligned}
            \Exp[Z_{t+1} \mid Z_t = z, \text{Case 2}] \cdot \Pr[\text{Case 2}]&= \int_{0}^{z} 2 \cdot (1 - z) \cdot ((1 - \lambda) \gamma + \lambda z) d \gamma  \\
            &=2\cdot (1 - z) \cdot \int_{0}^{z} ((1 - \lambda) \gamma + \lambda z) d \gamma \\
            &=2\cdot (1 - z) \cdot (\lambda z^2 + (1 - \lambda) \frac{1}{2} z^2)
        \end{aligned}
    \end{equation}
    Dividing by $\Pr[\text{Case 2}]=2z(1-z)$ gives
    \begin{equation}
        \begin{aligned}
            \Exp[Z_{t+1} \mid Z_t = z, \text{Case 2}] = \lambda z + \frac{1}{2} (1- \lambda)z = \frac{(1 + \lambda)}{2} \cdot z.
        \end{aligned}
    \end{equation}

    \sbpara{Case 3.} Both sampled points lie in $[X_t,Y_t]$.
    This has probability $z^2$.
    Let $\gamma_1=d(X_t,U)$ and $\gamma_2=d(Y_t,V)$.
    Then $d(U,V)=\abs{z-\gamma_1-\gamma_2}$ and
    $Z_{t+1}=(1-\lambda)d(U,V)+\lambda z$.
    The contribution of $d(U,V)$ is
    \begin{equation}
        \begin{aligned}
            \Exp[d(U, V) \mid Z_t = z, \text{Case 3}] \Pr[\text{Case 3}] &= \int_{0}^{z} \int_{0}^{z}  \abs{z - \gamma_1 - \gamma_2}  \, d\gamma_1 d\gamma_2 \\
            &= \int_{0}^z \left( \int_{0}^{z - \gamma_1} (z - \gamma_1 - \gamma_2) d \gamma_2 + \int_{z -\gamma_1}^{z} (\gamma_1 + \gamma_2 - z) d \gamma_2\right) d \gamma_1  \\
            &= \int_{0}^z \left(\frac{1}{2} (z - \gamma_1)^2 + \frac{1}{2} (\gamma_1)^2\right) d \gamma_1 \\
            &= \int_{0}^z  \left(\frac{1}{2} z^2 + \gamma_1^2 - z \gamma_1 \right) d \gamma_1 \\
            &= \frac{1}{2} z^3 + \frac{1}{3} z^3 - \frac{1}{2} z^3 = \frac{1}{3} z^3.
        \end{aligned}
    \end{equation}
    Since $\Pr[\text{Case 3}]=z^2$,
    \begin{equation}
        \begin{aligned}
            \Exp[d(U, V) \mid Z_t = z, \text{Case 3}]  = \frac{1}{3} z.
        \end{aligned}
    \end{equation}
    Applying Equation~\eqref{eq:contractionrule},
    \begin{equation}
        \begin{aligned}
            \Exp[Z_{t+1} \mid Z_t = z, \text{Case 3}] &= (1 - \lambda) \Exp[d(U, V) \mid Z_t = z, \text{Case 3}] + \lambda z  = \frac{1}{3} z + \frac{2}{3} \lambda z.
        \end{aligned}
    \end{equation}

    \sbpara{Case 4.} Both sampled points are outside $[X_t,Y_t]$ but on opposite sides.
    This has probability $(1-z)^2-\Pr[\text{Case 1}]\le (1-z)^2/2$.
    In this case the overlap length is $z$, so $Z_{t+1}=z$.
\end{proof}

\begin{figure}[t]
    \centering
    \begin{tikzpicture}[
        x=10cm,
        y=0.82cm,
        >=Stealth,
        axis/.style={->, line width=0.35pt},
        tick/.style={line width=0.3pt},
        xdens/.style={line width=0.9pt, color=blue!70!black},
        ydens/.style={line width=0.9pt, color=red!70!black},
        minor/.style={fill=green!45!black, fill opacity=0.32, draw=green!40!black, line width=0.5pt},
        lab/.style={font=\scriptsize},
        smalllab/.style={font=\tiny}
    ]
        \def\peak{3.333333}
        \def\midheight{0.555556}

        \fill[blue!14] (0,0) -- (0,\peak) -- (0.6,0) -- cycle;
        \fill[red!14] (0.4,0) -- (1,\peak) -- (1,0) -- cycle;
        \path[minor] (0.4,0) -- (0.5,\midheight) -- (0.6,0) -- cycle;

        \draw[axis] (-0.015,0) -- (1.055,0) node[lab, below] {$z$};
        \draw[axis] (0,0) -- (0,3.63) node[lab, left] {density};

        \draw[xdens] (0,\peak) -- (0.6,0);
        \draw[ydens] (0.4,0) -- (1,\peak);
        \draw[tick] (0,\peak) -- (-0.008,\peak) node[lab, left] {$\frac{10}{3}$};
        \draw[tick] (0,\midheight) -- (-0.008,\midheight) node[lab, left] {$\frac{5}{9}$};

        \foreach \xx/\name in {
            0/{$0$},
            0.4/{$0.4=\lambda$},
            0.5/{$1/2$},
            0.6/{$0.6=1-\lambda$},
            1/{$1$}
        }{
            \draw[tick] (\xx,0) -- (\xx,-0.09);
            \node[smalllab, below] at (\xx,-0.09) {\name};
        }

        \node[lab, blue!70!black, anchor=west] at (0.13,2.73) {$f_{X_1}(z)$ from $X_0=0$};
        \node[lab, red!70!black, anchor=west] at (0.58,2.40) {$f_{Y_1}(z)$ from $Y_0=1$};
        \node[lab, green!35!black, align=center] at (0.5,1.1)
            {overlap\\$m_{0.4}(z)$};

        \draw[<->, line width=0.35pt, color=green!35!black]
            (0.4,-0.5) -- node[lab, below] {$I_\lambda=[\lambda,1-\lambda]$} (0.6,-0.5);
    \end{tikzpicture}
    \caption{One-step minorization for the uniform chain with $\lambda=0.4$, starting from $X_0=0$ and $Y_0=1$.
    The two one-step densities are
    $f_{X_1}(z)=2(0.6-z)/0.6^2$ on $[0,0.6]$ and
    $f_{Y_1}(z)=2(z-0.4)/0.6^2$ on $[0.4,1]$.
    The shaded triangle is their common lower density on $I_\lambda=[0.4,0.6]$; its area is $\beta(0.4)=1/18$, which is the minorization mass in Lemma~\ref{lem:uniform-minorization}.}
    \label{fig:one-step-minorization}
\end{figure}
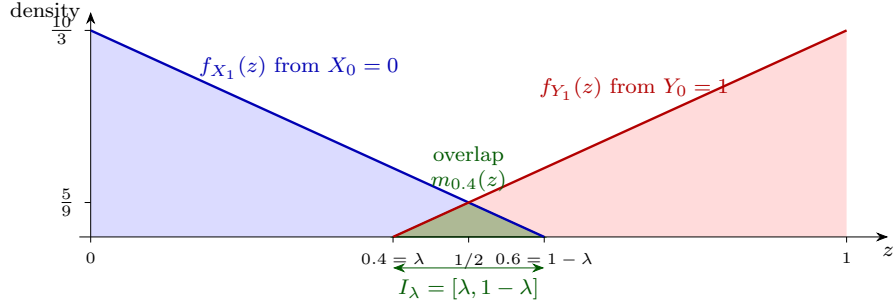

\sbpara{Contraction analysis based on one-step minorization.}
We now give a sharper but less general argument for total-variation distance. When $G$ is uniform and $\lambda<\frac{1}{2}$, every starting point has a common chance of moving into the interval $[\lambda,1-\lambda]$ in one step. Figure~\ref{fig:one-step-minorization} illustrates this common overlap. This overlap yields a one-step minorization and hence total-variation contraction.

\begin{lemma}
\label{lem:uniform-minorization}
Assume $G$ is uniform on $[0,1]$ and $0\le \lambda<\frac{1}{2}$.
Then there exists a probability measure $\nu_\lambda$ on $[0,1]$ such that
$K(x,\cdot)\ge \beta(\lambda)\nu_\lambda(\cdot)$ for every $x\in[0,1]$,
where $\beta(\lambda)=\frac{(1-2\lambda)^2}{2(1-\lambda)^2}$.
Consequently, for every pair of probability measures $\mu,\nu$ on $[0,1]$,
$\left\|\mu K - \nu K\right\|_{\mathrm{TV}} \le (1-\beta(\lambda))\|\mu-\nu\|_{\mathrm{TV}}.$
\end{lemma}

\begin{corollary}
    \label{cor:uniform-tv-mixing}
    Assume that $G$ is uniform on $[0, 1]$, and let $0 \leq \lambda < \frac{1}{2}$.
    Then $t^{\mathrm{TV}}_{\mathrm{mix}}(\varepsilon) \le \left\lceil
    \frac{1}{\beta(\lambda)}\log\frac{1}{\varepsilon}
    \right\rceil.$
\end{corollary}

\sbpara{Proof of Corollary~\ref{cor:coupling-concentration-expectation}}
\begin{proof}
    Combine the four cases in Lemma~\ref{lem:coupling-concentration}.
    The combined contribution of Cases 1 and 4 is at most
    \begin{equation}
        \label{eq:condition1}
    \begin{aligned}
    &\Exp[Z_{t+1} \mid Z_t = z, \text{Case 1}] \Pr[\text{Case 1}] +  \Exp[Z_{t+1} \mid Z_t = z, \text{Case 4}] \Pr[\text{Case 4}] \\
    &\stackrel{(a)}{\leq} (1 - z)^2 \frac{1}{2} \lambda z + (1 - z)^2\frac{1}{2} z \\
    &= \frac{(1-z)^2}{2} \cdot (\lambda z + z)
    \end{aligned}
    \end{equation}
    because $Z_{t+1}=\lambda z$ in Case 1, $Z_{t+1}=z$ in Case 4,
    $\Pr[\text{Case 1}]+\Pr[\text{Case 4}]=(1-z)^2$, and
    $\Pr[\text{Case 1}]\ge (1-z)^2/2$.
    \begin{equation}
        \begin{aligned}
            &\Exp[Z_{t+1} \mid Z_t = z] = \sum_{i = 1}^4\Exp[Z_{t+1} \mid Z_t = z, \text{Case $i$}] \Pr[\text{Case $i$}]\\
            &\leq \frac{(1-z)^2}{2} \cdot (\lambda z + z) +  2\cdot (1 - z) \cdot (\lambda z^2 + (1 - \lambda) \frac{1}{2} z^2) +  \left(\frac{1+2\lambda}{3}z  \right) z^2\\
            &= \frac{1 + \lambda}{2} z + \frac{\lambda - 1}{6} z^3 \leq \frac{1 + \lambda}{2} z.
        \end{aligned}
    \end{equation}
\end{proof}

\sbpara{Proof of Lemma~\ref{lem:uniform-minorization}}

\begin{proof}
    We use the uniform transition density in Equation~\eqref{eq:density}.
    When $\lambda < \frac{1}{2}$, for every $z\in I_\lambda=[\lambda,1-\lambda]$, $z$ belongs to the support interval $[\lambda x,\lambda x+(1-\lambda)]$ for every $x\in[0,1]$.
    If $x>z$, then
    \[
        k_x(z)=\frac{2(z-\lambda x)}{(1-\lambda)^2}
        \geq \frac{2(z-\lambda)}{(1-\lambda)^2},
    \]
    and if $x<z$, then
    \[
        k_x(z)=\frac{2(1-\lambda+\lambda x-z)}{(1-\lambda)^2}
        \geq
        \frac{2(1-\lambda-z)}{(1-\lambda)^2}.
    \]
    Hence, for every $x$,
    \[
    k_x(z)\ge m_\lambda(z)
    \coloneqq
    \frac{2}{(1-\lambda)^2}\min\{z-\lambda,\ 1-\lambda-z\}\ind_{[\lambda,1-\lambda]}(z).
    \]
    The mass of this common lower density is
    \[
    \beta(\lambda)=\int_0^1m_\lambda(z)\,dz
    =
    \frac{(1-2\lambda)^2}{2(1-\lambda)^2}.
    \]
    Let $\nu_\lambda$ be the probability measure with density $m_\lambda(z)/\beta(\lambda)$.
    Then $K(x,\cdot)\ge \beta(\lambda)\nu_\lambda(\cdot)$ for every $x$.

    Write $K(x,\cdot)=\beta(\lambda)\nu_\lambda(\cdot)+(1-\beta(\lambda))\widetilde K(x,\cdot)$.
    Since $\widetilde K$ is a Markov kernel, for every probability measures $\mu,\nu$,
    \[
    \|\mu K-\nu K\|_{\mathrm{TV}}
    =
    (1-\beta(\lambda))\|\mu\widetilde K-\nu\widetilde K\|_{\mathrm{TV}}
    \le
    (1-\beta(\lambda))\|\mu-\nu\|_{\mathrm{TV}}.
    \]
    \end{proof}

\sbpara{Proof of Corollary~\ref{cor:uniform-tv-mixing}}

\begin{proof}
    Lemma~\ref{lem:uniform-minorization} gives one-step total-variation contraction by the factor $1-\beta(\lambda)$.
    Iterating it and using stationarity of $\pi_{\lambda,G}$ gives
    \[
    \sup_{x\in[0,1]}\|K^t(x,\cdot)-\pi_{\lambda,G}\|_{\mathrm{TV}}
    \le
    (1-\beta(\lambda))^t.
    \]
    Since $(1-\beta)^t\le e^{-\beta t}$ for $\beta\in(0,1)$, the right-hand side is at most $\varepsilon$ whenever
    $t\ge \beta(\lambda)^{-1}\log(1/\varepsilon)$.
    This proves the stated bound.
\end{proof}

\section{Omitted proofs from Section~\ref{sec:worst-case}}
This section first proves Lemma~\ref{lemma:worst-case-tightness}; next it provides the ingredients for 
Lemma~\ref{lem:unanchored-line-benchmark}. 
The proof of Lemma~\ref{lem:unanchored-line-benchmark} is put to the end.

\subsection{Proof of Lemma~\ref{lemma:worst-case-tightness}}

\begin{proof}
    Consider the population with mass $p$ at $0$ and mass $q=1-p$ at $1$.
    Let $M_t=\Median(O_t,U_t,V_t)$, so
    \[
    O_{t+1}=(1-\lambda)M_t+\lambda O_t .
    \]
    With probability $p^2$, both sampled bliss points are $0$ and $M_t=0$.
    With probability $q^2$, both are $1$ and $M_t=1$.
    With probability $2pq$, the two sampled bliss points are split and $M_t=O_t$.

    Taking expectations at stationarity gives
    \[
    \Exp[O_\infty]
    =
    p^2\lambda\Exp[O_\infty]
    +q^2\bigl((1-\lambda)+\lambda\Exp[O_\infty]\bigr)
    +2pq\Exp[O_\infty].
    \]
    Since $\lambda<1$, this simplifies to
    \[
    \Exp[O_\infty]=\frac{q^2}{p^2+q^2},
    \]
    which is independent of $\lambda$.
    The social cost is
    \begin{equation}
        \begin{aligned}
            \Exp[\SCost(O_{\infty})]
            &= \Exp[pO_{\infty}+q(1-O_{\infty})] \\
            &= q+(p-q)\Exp[O_{\infty}]
            = \frac{pq}{p^2+q^2}.
        \end{aligned}
    \end{equation}

    Without loss of generality, assume $q\le p$, so the optimal outcome is $0$ and its social cost is $q$.
    The stationary distortion is therefore
    \[
    \frac{\Exp[\SCost(O_{\infty})]}{\SCost(0)}
    =
    \frac{p}{p^2+q^2}.
    \]
    Maximizing this expression over $p\in[1/2,1]$ gives $p=\frac{\sqrt2}{2}$ and $q=1-\frac{\sqrt2}{2}$, where the value is $\frac{1+\sqrt2}{2}$.
    This value is independent of $\lambda$.
\end{proof}

\subsection{Proof of Lemma~\ref{lem:social-cost-alternative}}
\begin{proof}
    Let $Y\sim G$ be independent of $O$.
    Since
    \[
        \abs{O-Y}
        =
        \int_0^1\left(\ind_{Y\le r<O}+\ind_{O\le r<Y}\right)\,dr,
    \]
    Fubini's theorem gives
    \[
    \begin{aligned}
        \Exp[\SCost(O)]
        &=\Exp[\abs{O-Y}]\\
        &=\int_0^1
        \left(
        \Pr(Y\le r<O)+\Pr(O\le r<Y)
        \right)\,dr\\
        &=\int_0^1
        \left(
        G(r)\Pr(O>r)+(1-G(r))\Pr(O\le r)
        \right)\,dr.
    \end{aligned}
    \]
    For deterministic $a$, this formulation becomes
    \[
    \SCost(a)=\int_0^a G(r)\,dr+\int_a^1(1-G(r))\,dr.
    \]
    Let $m_G$ be a population median, so $G(r)\le 1/2$ for $r<m_G$ and $G(r)\ge 1/2$ for $r>m_G$, except possibly on a median interval where either choice gives the same value.
    For any $a\ge m_G$,
    \[
    \SCost(a)-\SCost(m_G)=\int_{m_G}^a(2G(r)-1)\,dr\ge0,
    \]
    and for any $a\le m_G$,
    \[
    \SCost(a)-\SCost(m_G)=\int_a^{m_G}(1-2G(r))\,dr\ge0.
    \]
    Thus the minimum deterministic social cost is attained at any population median.
    Finally,
    \[
    \SCost(m_G)
    =
    \int_0^{m_G}G(r)\,dr+\int_{m_G}^1(1-G(r))\,dr
    =
    \int_0^1\min\{G(r),1-G(r)\}\,dr.
    \]
\end{proof}

\subsection{Proof of Lemma~\ref{lem:fain-general-law}}

\begin{proof}
    The transition CDF is proved by the same argument as in Proposition~\ref{prop:general-transition-cdf}.
    For every $x,z\in[0,1]$,
    \[
    \Pr(X_{t+1}\le z\mid X_t=x)
    =
    \begin{cases}
    G(z)^2, & z<x,\\[2mm]
    1-\bigl(1-G(z)\bigr)^2, & z\ge x.
    \end{cases}
    \]
    For the stationary CDF, fix $z\in[0,1]$ and write $F_t(z)=\Pr(X_t\le z)$.
    Conditioning on whether $X_t\le z$ gives
    \[
    \begin{aligned}
    F_{t+1}(z)
    &=
    F_t(z)\bigl[1-\bigl(1-G(z)\bigr)^2\bigr]
    +(1-F_t(z))G(z)^2 \\
    &=
    G(z)^2+2G(z)(1-G(z))F_t(z).
    \end{aligned}
    \]
    At stationarity, $F_{t+1}(z)=F_t(z)=\Pi_G(z)$, which implies
    \[
    \Pi_G(z)
    =
    \frac{G(z)^2}{1-2G(z)(1-G(z))}
    =
    \frac{G(z)^2}{G(z)^2+(1-G(z))^2}.
    \]
\end{proof}

\subsection{Proof of Lemma~\ref{lem:unanchored-line-benchmark}}

\begin{proof}
    By Lemma~\ref{lem:social-cost-alternative}, the social cost of the unanchored stationary distribution can be written as
    \begin{equation}
        \begin{aligned}
        \Exp[\SCost(X_{\infty})] &= \int_0^1 \left[ G(r)\Pr(X_{\infty}>r)+ (1 - G(r))\Pr(X_{\infty}\le r)\right]d\,r \\
        &\stackrel{(a)}{=}\int_{0}^1 \frac{G(r) (1 - G(r))}{G(r)^2 + (1 - G(r))^2} d\,r \\
        &\stackrel{(b)}{\leq} \int_{0}^1 \frac{\sqrt{2} + 1}{2} \min \{G(r), 1 - G(r)\} d\,r \\
        &\stackrel{(c)}{=} \frac{\sqrt{2} + 1}{2} \SCost(m_G).
        \end{aligned}
    \end{equation}

    Step $(a)$ substitutes $\Pr(X_{\infty}\le r)=\Pi_G(r)$ and $\Pr(X_{\infty}>r)=1-\Pi_G(r)$ using Lemma~\ref{lem:fain-general-law}.
    Step $(b)$ follows from
    \[
    \frac{G(r)(1-G(r))}{G(r)^2+(1-G(r))^2}
    \le
    \frac{\sqrt2+1}{2}\min\{G(r),1-G(r)\},
    \]
    which is equivalent to
    $\max_{0\le p\le 1/2}(1-p)/(p^2+(1-p)^2)\le(\sqrt2+1)/2$.
    Step $(c)$ applies Lemma~\ref{lem:social-cost-alternative}, which gives
    $\SCost(m_G)=\int_0^1\min\{G(r),1-G(r)\}\,dr$ for any population median $m_G$.
    Since every population median minimizes deterministic social cost, the lemma follows.

        \end{proof}

\section{Omitted proofs from Section~\ref{sec:general-distortion}}
\label{sec:general-distortion-omitted}

\subsection{Proof of Lemma~\ref{lem:deliberative-fixed-point}}

\begin{proof}
        When $U$ and $V$ lie on different sides of $x$, $\med(x,U,V)=x$.
        Thus only pairs on the same side of $x$ contribute to the expected movement.
        We can write $b_G(x)$ as

        \begin{equation}
            \begin{aligned}
                b_G(x) &= \Exp[\med(x, U, V) - x] \\
                &\stackrel{(a)}{=} \Exp[(\med(x, U, V) - x)_{+}] - \Exp[(x -\med(x, U, V))_{+}] \\
                &= \Exp[\int_{x}^{1}\ind_{\med(x, U, V) \geq r} d\, r] - \Exp[\int_{0}^x \ind_{\med(x, U, V) \leq r} d\, r] \\
                &= \int_{x}^{1} \Pr(U > r, V > r) d\, r - \int_{0}^x \Pr(U < r, V < r) d\, r  \\
                &= \int_{x}^1 (1 - G(r))^2 d\,r - \int_{0}^x G(r)^2 d\, r.
            \end{aligned}
        \end{equation}
        Here $(\cdot)_+=\max\{\cdot,0\}$, and $(a)$ uses the identity $y=y_+-(-y)_+$.
        The last equality is unaffected by the choice of strict or weak threshold inequalities, because the left limits of $G$ differ from $G$ only at countably many points.
        For $0\le x<y\le1$,
        \[
        b_G(y)-b_G(x)
        =
        -\int_x^y\left(G(r)^2+(1-G(r))^2\right)\,dr
        \le -\frac{1}{2}(y-x).
        \]
        Thus $b_G$ is continuous and strictly decreasing.
        Moreover,
        \[
        b_G(0)=\int_0^1(1-G(r))^2\,dr\ge0,
        \qquad
        b_G(1)=-\int_0^1G(r)^2\,dr\le0.
        \]
        By the intermediate value theorem, there exists $\theta_G\in[0,1]$ with $b_G(\theta_G)=0$.
        Strict monotonicity gives uniqueness.
\end{proof}

\subsection{Proof of Theorem~\ref{thm:general-stationary-concentration}}

\begin{proof}
        Let $O_{x,\lambda}$ be the one-step outcome from $O_t=x$ with anchoring parameter $\lambda$.
        Then $O_{x,\lambda}=x+(1-\lambda)(O_{x,0}-x)$.
        We bound the conditional second moment around $\theta_G$.

        \begin{equation}
            \begin{aligned}
                \Exp[(O_{t+1}-\theta_G)^2\mid O_t=x]
                &=\Exp[\left((x - \theta_G) + (1 - \lambda) (O_{x, 0} - x)\right)^2] \\
                &\stackrel{(a)}{=}(x-\theta_G)^2+2(1-\lambda)(x-\theta_G)b_G(x)
                +(1-\lambda)^2\Exp[(O_{x, 0}-x)^2] \\
                &\stackrel{(b)}{\le} (x-\theta_G)^2-(1-\lambda)(x-\theta_G)^2+(1-\lambda)^2 \\
                &= \lambda (x-\theta_G)^2+(1-\lambda)^2.
            \end{aligned}
        \end{equation}
        where $(a)$ follows from linearity of expectation and $b_G(x)=\Exp[O_{x,0}-x]$.
        For $(b)$, Lemma~\ref{lem:deliberative-fixed-point} gives
        \[
        b_G(x)-b_G(\theta_G)
        =
        -\int_{\theta_G}^{x}
        \left(G(r)^2+(1-G(r))^2\right)\,dr.
        \]
        Since $b_G(\theta_G)=0$, this implies
        \[
        (x-\theta_G)b_G(x)
        \le
        -\frac{1}{2}(x-\theta_G)^2.
        \]
        We also use $(O_{x,0}-x)^2\le1$.
        Taking expectations under the stationary distribution gives
        \[
        \Exp[(O_{\lambda}-\theta_G)^2] \leq \lambda\Exp[(O_{\lambda}-\theta_G)^2]+(1-\lambda)^2.
        \]
        Solving yields $\Exp[(O_{\lambda}-\theta_G)^2]\leq 1-\lambda$, where $O_{\lambda} \sim \pi_{\lambda,G}$.
\end{proof}

\subsection{Proof of Lemma~\ref{lem:centered-fixed-point}}
\begin{proof}
    By the integral formula for $b_G$ from Lemma~\ref{lem:deliberative-fixed-point},
    \[
    b_G\!\left(\frac{1}{2}\right)
    =
    \int_{1/2}^1 (1-G(r))^2\,dr
    -
    \int_0^{1/2}G(r)^2\,dr.
    \]
    In the first integral, substitute $r=\frac{1}{2}+s$.
    Since $G$ is centered, $1-G(\frac{1}{2}+s)=G(\frac{1}{2}-s)$.
    Therefore
    \[
    \int_{1/2}^1 (1-G(r))^2\,dr
    =
    \int_0^{1/2}G\!\left(\frac{1}{2}-s\right)^2\,ds
    =
    \int_0^{1/2}G(u)^2\,du.
    \]
    Hence $b_G(1/2)=0$.
    Lemma~\ref{lem:deliberative-fixed-point} gives uniqueness, so $\theta_G=1/2$.
\end{proof}

\subsection{Proof of Theorem~\ref{thm:distortion-uniform}}
\begin{proof}
	        The upper and lower bounds follow directly from Corollary~\ref{cor:upperQ}.
	        We now compute the exact value when $\lambda=0$.
        By Lemma~\ref{lem:fain-general-law}, the stationary CDF for the uniform unanchored chain is
        \[
        \Pr(O_\infty\le r) =
        \frac{r^2}{r^2+(1-r)^2}.
        \]
        Applying Lemma~\ref{lem:social-cost-alternative} with $G(r)=r$ gives
        \[
        \begin{aligned}
        \Exp[\SCost(O_\infty)]
        &=
        \int_0^1
        \left(r\Pr(O_\infty>r)+(1-r)\Pr(O_\infty\le r)\right)\,dr\\
        &=
	        \int_0^1
	        \frac{r(1-r)}{r^2+(1-r)^2}\,dr \\
	        &=
	        \frac{1}{4}\int_{-1}^1 \frac{1-s^2}{1+s^2}\,ds
	        =
	        \frac{\pi-2}{4},
	        \end{aligned}
	        \]
	        where the penultimate equality uses the substitution $s=2r-1$.
	        Since $\SCost(1/2)=1/4$, we have the distortion $\pi-2$.
\end{proof}

\subsection{Proof of Lemma~\ref{lem:social-cost-uniform}}
\begin{proof}

    When the distribution is uniform, for any fixed outcome $O_t\in[0,1]$ we have
    \[
        \SCost(O_t)
        = \int_0^{O_t} (O_t-x)\,dx + \int_{O_t}^1 (x-O_t)\,dx
        = \frac{O_t^2}{2} + \frac{(1-O_t)^2}{2}
        = \frac{1}{4} + \left(O_t-\frac{1}{2}\right)^2.
    \]
    Since $\SCost(O^*)=\SCost(1/2)=1/4$, using $(O_t-1/2)^2=\Qcal_t^2$ in the above equation gives

        \begin{equation}
        \frac{\Exp[\SCost(O_t)]}{\SCost(O^*)}
        = 1 + 4\Exp[\Qcal_t^2].
        \end{equation}

\end{proof}

\subsection{Proof of Lemma~\ref{lem:centered-second-moment}}

\begin{proof}
    We first recall the transition density for the uniform distribution from Equation~\eqref{eq:density}:

    \begin{equation*}
    k_x(z)=
    \begin{cases}
    \dfrac{2({z-\lambda x})}{(1-\lambda)^2}, & \lambda x<z<x,\\[2mm]
    \dfrac{2(1-\lambda+\lambda x-z)}{(1-\lambda)^2}, & x<z<\lambda x+(1-\lambda),\\[2mm]
    0, & \text{otherwise},
    \end{cases}
    \end{equation*}
    plus the atom $2x(1-x)\delta_x$, where $\delta_x$ is the Dirac mass at $x$.

    Fix $x\in[0,1]$. By the transition density,
    \begin{equation}
    \begin{aligned}
    \Exp[\Qcal_{t+1}^2\mid O_t=x]
    &=
    2x(1-x)\left(x-\frac{1}{2}\right)^2 \\
    &\quad+
    \frac{2}{(1-\lambda)^2}\int_{\lambda x}^{x}\left(z-\frac{1}{2}\right)^2(z-\lambda x)\,dz \\
    &\quad+
    \frac{2}{(1-\lambda)^2}\int_{x}^{1-\lambda+\lambda x}\left(z-\frac{1}{2}\right)^2(1-\lambda+\lambda x-z)\,dz.
    \end{aligned}
    \end{equation}

    By reflection symmetry, the conditional second moment is the same at $O_t=\frac{1}{2}+y$ and $O_t=\frac{1}{2}-y$.
    Hence, for $0\le y\le\frac{1}{2}$, conditioning on $\Qcal_t=y$ gives this same value.
    Substituting $x = \frac{1}{2} +y$ gives
    \begin{equation}
    \begin{aligned}
    \Exp[\Qcal_{t+1}^2\mid \Qcal_t=y]
    &=
    \left(\frac{1}{2} y^2-2y^4\right) \\
    &\quad+
    \frac{2}{(1-\lambda)^2}
    \int_{\frac{\lambda}{2}+\lambda y}^{\frac{1}{2}+y}
    \left(z-\frac{1}{2}\right)^2
    \left(z-\frac{\lambda}{2}-\lambda y\right)\,dz \\
    &\quad+
    \frac{2}{(1-\lambda)^2}
    \int_{\frac{1}{2}+y}^{\,1-\frac{\lambda}{2}+\lambda y}
    \left(z-\frac{1}{2}\right)^2
    \left(1-\frac{\lambda}{2}+\lambda y-z\right)\,dz.
    \end{aligned}
    \end{equation}

    Now set $u=z-\frac{1}{2}$, so $z=u+\frac{1}{2}$.
    The integration limits become $u=y$ when $z=\frac{1}{2}+y$,
    $u=-\frac{1-\lambda}{2}+\lambda y$ when $z=\frac{\lambda}{2}+\lambda y$,
    and $u=\frac{1-\lambda}{2}+\lambda y$ when $z=1-\frac{\lambda}{2}+\lambda y$.

    \[
    \begin{aligned}
    \Exp[\Qcal_{t+1}^2\mid \Qcal_t=y]
    &=
    \frac{1}{2} y^2-2y^4 \\
    &\quad+
    \frac{2}{(1-\lambda)^2}
    \int_{-\frac{1-\lambda}{2}+\lambda y}^{y}
    u^2\left(u+\frac{1-\lambda}{2}-\lambda y\right)\,du \\
    &\quad+
    \frac{2}{(1-\lambda)^2}
    \int_{y}^{\frac{1-\lambda}{2}+\lambda y}
    u^2\left(\frac{1-\lambda}{2}+\lambda y-u\right)\,du.
    \end{aligned}
    \]

    We now compute the two integral terms.

    \begin{equation}
        \label{eq:subint1}
        \begin{aligned}
        &\int_{-\frac{1-\lambda}{2}+\lambda y}^{y} u^2\left(u+\frac{1-\lambda}{2}-\lambda y\right)\,du = \left(\frac{1}{4}u^4 + \frac{1}{3}(\frac{1-\lambda}{2} - \lambda y) u^3\right) \big|_{-\frac{1-\lambda}{2} + \lambda y}^y \\
        &\qquad\qquad= \left(\frac{1}{4}y^4 + \frac{1}{3}(\frac{1-\lambda}{2} - \lambda y) y^3\right) - \left(\frac{1}{4}(-\frac{1-\lambda}{2} + \lambda y)^4 + \frac{1}{3}(\frac{1-\lambda}{2} - \lambda y) (-\frac{1-\lambda}{2} + \lambda y)^3\right) \\
        &\qquad\qquad= \frac{1}{4}y^4 + \frac{1}{3}(\frac{1-\lambda}{2} - \lambda y) y^3 + \frac{1}{12}(-\frac{1-\lambda}{2} + \lambda y)^4
        \end{aligned}
    \end{equation}

    \begin{equation}
        \label{eq:subint2}
        \begin{aligned}
         &\int_{y}^{\frac{1-\lambda}{2}+\lambda y} u^2\left(\frac{1-\lambda}{2}+\lambda y-u\right)\,du = \left(\frac{1}{3} u^3 (\frac{1-\lambda}{2} + \lambda y) - \frac{1}{4} u^4\right) \Big|_{y}^{\frac{1 - \lambda}{2} + \lambda y} \\
         &\qquad\qquad = \left(\frac{1}{3} (\frac{1 - \lambda}{2} + \lambda y)^3 (\frac{1-\lambda}{2} + \lambda y) - \frac{1}{4} (\frac{1 - \lambda}{2} + \lambda y)^4\right) - \left(\frac{1}{3} y^3 (\frac{1-\lambda}{2} + \lambda y) - \frac{1}{4} y^4\right) \\
         &\qquad\qquad = \frac{1}{4} y^4 - \frac{1}{3}y^3(\frac{1 - \lambda}{2} + \lambda y) + \frac{1}{12} (\frac{1 - \lambda}{2} + \lambda y)^4
        \end{aligned}
    \end{equation}

    Adding Equations~\eqref{eq:subint1} and~\eqref{eq:subint2} gives
    \begin{equation}
        \label{eq:summing}
        \begin{aligned}
        &\frac{1}{2} y^4 - \frac{2\lambda}{3}y^4 + \frac{1}{12} \left((-\frac{1-\lambda}{2} + \lambda y)^4 + (\frac{1-\lambda}{2} + \lambda y)^4 \right) \\
        &\qquad\qquad= (\frac{1}{2} - \frac{2\lambda}{3}) y^4 + \frac{1}{12} \left(2 (\frac{1-\lambda}{2})^4 + 12 (\frac{1-\lambda}{2})^2 (\lambda y)^2 + 2 (\lambda y)^4\right) \\
        &\qquad\qquad= (\frac{1}{2} - \frac{2\lambda}{3} + \frac{1}{6} \lambda^4)y^4 +  \frac{(1 - \lambda)^2}{4} \lambda^2 y^2 + \frac{1}{6} (\frac{1 - \lambda}{2})^4
        \end{aligned}
    \end{equation}

    Combining these terms,
    \begin{equation}
        \begin{aligned}
        \Exp[\Qcal_{t+1}^2\mid \Qcal_t=y]
        &=
        \left(\frac{1}{2} y^2-2y^4\right) + \frac{2}{(1 -\lambda)^2} \left((\frac{1}{2} - \frac{2\lambda}{3} + \frac{1}{6} \lambda^4)y^4 +  \frac{(1 - \lambda)^2}{4} \lambda^2 y^2 + \frac{1}{6} (\frac{1 - \lambda}{2})^4  \right) \\
        &= \frac{(1 -\lambda)^2}{48} + \frac{1 + \lambda^2}{2} y^2 - \frac{(1 - \lambda)(\lambda + 3)}{3} y^4
        \end{aligned}
        \end{equation}
\end{proof}

\subsection{Proof of Corollary~\ref{cor:upperQ}}
\begin{proof}
        \sbpara{Upper bound.}
        Taking expectations in Lemma~\ref{lem:centered-second-moment} gives
        $\Exp[\Qcal_{t+1}^2] \leq \frac{(1 -\lambda)^2}{48} + \frac{1 + \lambda^2}{2} \Exp[\Qcal_{t}^2]$.
        Let $a \coloneqq \frac{(1 - \lambda)^2}{48}$ and $b = \frac{1 + \lambda^2}{2}$.
        Iterating the recursion gives
        \begin{equation}
            \begin{aligned}
                \Exp[\Qcal_{t+1}^2] &\leq a + b \Exp[\Qcal_t^2] \\
                &=a(1 + b + \ldots + b^{t-1}) + b^t \Exp[\Qcal_1^2] \\
                &=a\cdot\frac{1 - b^t}{1 - b} + b^t \Exp[\Qcal_1^2] \\
                &=\frac{a}{1 -b} + b^t \left(\Exp[\Qcal_1^2] - \frac{a}{1-b}\right) \\
                &= \frac{1-\lambda}{24(1+\lambda)} + \left(\frac{1 + \lambda^2}{2}\right)^t \left(\Exp[\Qcal_1^2] - \frac{1-\lambda}{24(1 + \lambda)}\right).
            \end{aligned}
        \end{equation}
        Since $\Exp[\Qcal_1^2]\le\frac14$,
        \[
        \Exp[\Qcal_1^2]-\frac{1-\lambda}{24(1+\lambda)}
        \le
        \frac{5+7\lambda}{24(1+\lambda)}
        \]
        This implies the stated finite-time upper bound; taking $t\to\infty$ gives the stationary upper bound.

        \sbpara{Lower bound.}
        Since $\Qcal_t^4\le\frac14\Qcal_t^2$, Lemma~\ref{lem:centered-second-moment} also gives
        $\Exp[\Qcal_{t+1}^2] \geq \frac{(1 -\lambda)^2}{48} + \frac{7\lambda^2+2\lambda+3}{12} \Exp[\Qcal_{t}^2]$.
        Let $a \coloneqq \frac{(1 - \lambda)^2}{48}$ and $b = \frac{7\lambda^2+2\lambda+3}{12}$.
        Since $\Exp[\Qcal_1^2]\ge0$, iterating gives
        \begin{equation}
            \begin{aligned}
                \Exp[\Qcal_{t+1}^2] &\geq a + b \Exp[\Qcal_t^2] \\
                &=a(1 + b + \ldots + b^{t-1}) + b^t \Exp[\Qcal_1^2] \\
                &=a\cdot\frac{1 - b^t}{1 - b} + b^t \Exp[\Qcal_1^2] \\
                &\geq\frac{a}{1 -b} + b^t \left( - \frac{a}{1-b}\right) \\
                &= \frac{1-\lambda}{4(9+7\lambda)}
                \left[1-\left(\frac{7\lambda^2+2\lambda+3}{12}\right)^t\right]
            \end{aligned}
        \end{equation}

\end{proof}

\subsection{Proof of Lemma~\ref{lem:centered-second-moment-general-g}}

\begin{proof}
    By symmetry, for $0\le q\le\frac{1}{2}$,
    $\Exp[\Qcal_{t+1}^2\mid O_t=\frac{1}{2}+q] = \Exp[\Qcal_{t+1}^2\mid O_t=\frac{1}{2}-q]$.
    Hence, this common value is $\Exp[\Qcal_{t+1}^2\mid \Qcal_t=q]$.
    Write $x_q=\frac{1}{2}+q$ and condition on $O_t=x_q$.
    By Proposition~\ref{prop:general-transition-density}, the conditional distribution has an atom at $O_t$ and a continuous part on the two sides of $O_t$.
    For the display below, set
    $u_z=(z-\lambda x_q)/(1-\lambda)$.
    Therefore
    \[
    \begin{aligned}
    \Exp[\Qcal_{t+1}^2\mid O_t=x_q]
    &=
    2G(x_q)\bigl(1-G(x_q)\bigr)q^2\\
    &+
    \frac{2}{1-\lambda}
    \int_{\lambda x_q}^{x_q}
    \left(z-\frac{1}{2}\right)^2
    G(u_z)g(u_z)\,dz\\
    &+
    \frac{2}{1-\lambda}
    \int_{x_q}^{1-\lambda+\lambda x_q}
    \left(z-\frac{1}{2}\right)^2
    \bigl(1-G(u_z)\bigr)g(u_z)\,dz.
    \end{aligned}
    \]
    Substitute $u=\frac{z-\lambda x_q}{1-\lambda}$, so
    $z-\frac{1}{2}=\lambda q+(1-\lambda)\left(u-\frac{1}{2}\right)$.
    Then
    \begin{equation}
    \begin{aligned}
    \Exp[\Qcal_{t+1}^2\mid \Qcal_t=q]
    &=
    2G\!\left(\frac{1}{2}+q\right)
    \left(1-G\!\left(\frac{1}{2}+q\right)\right)q^2\\
    &\quad+
    2\int_0^{\frac{1}{2}+q}
    \left[\lambda q+(1-\lambda)\left(u-\frac{1}{2}\right)\right]^2
    G(u)g(u)\,du\\
    &\quad+
    2\int_{\frac{1}{2}+q}^{1}
    \left[\lambda q+(1-\lambda)\left(u-\frac{1}{2}\right)\right]^2
    \left(1-G(u)\right)g(u)\,du.
    \end{aligned}
    \end{equation}
    For the first integral, split at $1/2$ and substitute $u=\frac{1}{2}-s$ on the left half and $u=\frac{1}{2}+v$ on the right half.
    \begin{equation}
    \begin{aligned}
    &2\int_0^{\frac{1}{2}+q} \left[\lambda q+(1-\lambda)\left(u-\frac{1}{2}\right)\right]^2 G(u)g(u)\,du \\
    &\quad=
    2\int_0^{\frac{1}{2}} \left[\lambda q+(1-\lambda)\left(u-\frac{1}{2}\right)\right]^2 G(u)g(u)\,du +
    2\int_{\frac{1}{2}}^{\frac{1}{2}+q} \left[\lambda q+(1-\lambda)\left(u-\frac{1}{2}\right)\right]^2 G(u)g(u)\,du\\
    &\quad= 2\int_0^{\frac{1}{2}} \left[\lambda q-(1-\lambda)s\right]^2 G(\frac{1}{2}-s)g(\frac{1}{2}-s)\,ds + 2\int_0^{q} \left[\lambda q+(1-\lambda)v\right]^2 G(\frac{1}{2}+v)g(\frac{1}{2}+v)\,dv.
    \end{aligned}
    \end{equation}

    For the second integral, we substitute $u=\frac{1}{2}+s$. Using the symmetry of $G$, $1-G(\frac{1}{2}+s)=G(\frac{1}{2}-s)$.
    Therefore
    \begin{equation}
    \begin{aligned}
    &2\int_{\frac{1}{2}+q}^{1}
    \left[\lambda q+(1-\lambda)\left(u-\frac{1}{2}\right)\right]^2
    \left(1-G(u)\right)g(u)\,du\\
    &\qquad=
    2\int_q^{1/2}
    \left(\lambda q+(1-\lambda)s\right)^2
    G\!\left(\frac{1}{2}-s\right)
    g\!\left(\frac{1}{2}+s\right)\,ds.
    \end{aligned}
    \end{equation}

    Since $g$ is symmetric around $1/2$, $g(\frac{1}{2}-s)=g(\frac{1}{2}+s)$.
    We split the integral into the intervals $[0,q]$ and $[q,1/2]$.
    Putting these pieces together, we have
    \begin{equation}
    \begin{aligned}
    \Exp[\Qcal_{t+1}^2\mid \Qcal_t=q]
    &=
    2G\!\left(\frac{1}{2}+q\right)
    \left(1-G\!\left(\frac{1}{2}+q\right)\right)q^2\\
    &\quad+
    2\int_0^q
    \left(\lambda q-(1-\lambda)s\right)^2
    G\!\left(\frac{1}{2}-s\right)
    g\!\left(\frac{1}{2}+s\right)\,ds\\
    &\quad+
    2\int_q^{1/2}
    \left(\lambda q-(1-\lambda)s\right)^2
    G\!\left(\frac{1}{2}-s\right)
    g\!\left(\frac{1}{2}+s\right)\,ds\\
    &\quad+
    2\int_0^q
    \left(\lambda q+(1-\lambda)s\right)^2
    G\!\left(\frac{1}{2}+s\right)
    g\!\left(\frac{1}{2}+s\right)\,ds\\
    &\quad+
    2\int_q^{1/2}
    \left(\lambda q+(1-\lambda)s\right)^2
    G\!\left(\frac{1}{2}-s\right)
    g\!\left(\frac{1}{2}+s\right)\,ds.
    \end{aligned}
    \end{equation}

    Expanding the four squared terms, and grouping the
    $\lambda^2q^2$, $(1-\lambda)^2s^2$, and $\lambda(1-\lambda)qs$ parts, gives
    \[
    \begin{aligned}
    \Exp[\Qcal_{t+1}^2\mid \Qcal_t=q]
    &=
    2G\!\left(\frac{1}{2}+q\right)
    \left(1-G\!\left(\frac{1}{2}+q\right)\right)q^2\\
    &\quad+
    2\lambda^2q^2
    \left[
    \int_0^q g\!\left(\frac{1}{2}+s\right)\,ds
    +
    2\int_q^{1/2}
    G\!\left(\frac{1}{2}-s\right)g\!\left(\frac{1}{2}+s\right)\,ds
    \right]\\
    &\quad+
    2(1-\lambda)^2
    \left[\int_0^q s^2g\!\left(\frac{1}{2}+s\right)\,ds + 2\int_q^{1/2} s^2G\!\left(\frac{1}{2}-s\right)g\!\left(\frac{1}{2}+s\right)\,ds\right]\\
    &\quad+
    4\lambda(1-\lambda)q
    \int_0^q
    s\left[
    G\!\left(\frac{1}{2}+s\right)-G\!\left(\frac{1}{2}-s\right)
    \right]
    g\!\left(\frac{1}{2}+s\right)\,ds.
    \end{aligned}
    \]
    Here we used $G(\frac{1}{2}-s)+G(\frac{1}{2}+s)=1$ on the interval $[0,q]$, and the
    cross terms on $[q,1/2]$ cancel. Next, we simplify the terms in each group.

    \sbpara{A. $2G\!\left(\frac{1}{2}+q\right)\left(1-G\!\left(\frac{1}{2}+q\right)\right)q^2$.}
    \begin{equation}
    \begin{aligned}
    2G\!\left(\frac{1}{2}+q\right)\left(1-G\!\left(\frac{1}{2}+q\right)\right)q^2
    &=2\left(\frac{1}{4} - \left(G\!\left(\frac{1}{2} + q\right) - \frac{1}{2}\right)^2\right)q^2 \\
    &= \frac{1}{2}q^2 - 2q^2\left(G\!\left(\frac{1}{2} + q\right) - \frac{1}{2}\right)^2 \\
    \end{aligned}
    \end{equation}

    \sbpara{B. The $2\lambda^2q^2$ group.}
    We again apply the symmetry property of $g$ to simplify the integral.
    We have $\int_0^q g\!\left(\frac{1}{2}+s\right)\,ds=G\!\left(\frac{1}{2}+q\right)-\frac{1}{2}$
    and, using $g(\frac{1}{2}+s)=g(\frac{1}{2}-s)$, $\int_q^{1/2} G\!\left(\frac{1}{2}-s\right)g\!\left(\frac{1}{2}+s\right)\,ds= \frac{1}{2}\left(1-G\!\left(\frac{1}{2}+q\right)\right)^2$.
    \begin{equation}
    \begin{aligned}
        &2\lambda^2q^2
        \left[
        \int_0^q g\!\left(\frac{1}{2}+s\right)\,ds + 2\int_q^{1/2} G\!\left(\frac{1}{2}-s\right)g\!\left(\frac{1}{2}+s\right)\,ds
        \right] \\
        &\quad=
        2\lambda^2q^2
        \left[
        G\!\left(\frac{1}{2}+q\right)-\frac{1}{2} + \left(1-G\!\left(\frac{1}{2}+q\right)\right)^2
        \right] \\
        &\quad= 2\lambda^2q^2 \left[G(\frac{1}{2} + q) - \frac{1}{2} + 1 + G(\frac{1}{2} +q)^2 -2G(\frac{1}{2} + q)\right] \\
        &\quad = 2\lambda^2 q^2 \left[(G(\frac{1}{2} +q) - \frac{1}{2})^2 + \frac{1}{4} \right] \\
        &\quad = 2\lambda^2 q^2 (G(\frac{1}{2} +q) - \frac{1}{2})^2 + \frac{1}{2} \lambda^2 q^2
    \end{aligned}
    \end{equation}

    \sbpara{C. The $(1-\lambda)^2s^2$ group.} We use $s\le q$ on $[0,q]$ to derive an upper bound.
    \begin{equation}
        \begin{aligned}
            &2(1-\lambda)^2 \left[\int_0^q s^2g\!\left(\frac{1}{2}+s\right)\,ds + 2\int_q^{1/2} s^2G\!\left(\frac{1}{2}-s\right)g\!\left(\frac{1}{2}+s\right)\,ds\right] \\
            &\quad= 2(1-\lambda)^2 \left[\int_0^q s^2 (1- 2G(\frac{1}{2} - s))g\!\left(\frac{1}{2}+s\right)\,ds + 2\int_0^{1/2} s^2G\!\left(\frac{1}{2}-s\right)g\!\left(\frac{1}{2}+s\right)\,ds\right]  \\
            &\quad= 4(1-\lambda)^2 \int_0^{1/2} s^2G\!\left(\frac{1}{2}-s\right)g\!\left(\frac{1}{2}+s\right)\,ds + 2 (1-\lambda)^2\int_0^q s^2 (1- 2G(\frac{1}{2} - s))g\!\left(\frac{1}{2}+s\right)\,ds \\
            &\quad= 4 (1-\lambda)^2 \int_0^{1/2}\left(\frac{1}{2}-u\right)^2G(u)g(u)\,du + 2 (1-\lambda)^2\int_0^q s^2 (1- 2G(\frac{1}{2} - s))g\!\left(\frac{1}{2}+s\right)\,ds \\
            &\quad \leq 4 (1-\lambda)^2 \int_0^{1/2}\left(\frac{1}{2}-u\right)^2G(u)g(u)\,du + 4 q^2(1-\lambda)^2\int_0^q (G(\frac{1}{2} + s) - \frac{1}{2})g\!\left(\frac{1}{2}+s\right)\,ds \\
            &\quad \leq 4 (1-\lambda)^2 \int_0^{1/2}\left(\frac{1}{2}-u\right)^2G(u)g(u)\,du + 2q^2(1-\lambda)^2 \left(G(\frac{1}{2} + q) - \frac{1}{2}\right)^2
        \end{aligned}
    \end{equation}

    \sbpara{D. The $4 \lambda (1 - \lambda) q$ group.} Again using $s\le q$ on $[0,q]$,
    \begin{equation}
        \begin{aligned}
            &4\lambda(1-\lambda)q \int_0^q s\left[G\!\left(\frac{1}{2}+s\right)-G\!\left(\frac{1}{2}-s\right)\right] g\!\left(\frac{1}{2}+s\right)\,ds \\
            &\quad=8\lambda(1-\lambda)q \int_0^q s\left[G\!\left(\frac{1}{2}+s\right)-\frac{1}{2}\right] g\!\left(\frac{1}{2}+s\right)\,ds \\
            &\quad\le 8\lambda(1-\lambda)q^2 \int_{\frac{1}{2}}^{q + \frac{1}{2}} \left(G(u)-\frac{1}{2}\right) g(u)\,du \\
            &\quad=4\lambda(1 -\lambda)q^2 \left(G(\frac{1}{2} + q) - \frac{1}{2}\right)^2
        \end{aligned}
    \end{equation}

    Adding the four groups A--D and reorganizing the terms, we obtain

    \begin{equation}
        \begin{aligned}
            \Exp[\Qcal_{t+1}^2\mid \Qcal_t=q] &\leq \frac{1}{2}q^2 - 2q^2\left(G\!\left(\frac{1}{2} + q\right) - \frac{1}{2}\right)^2 \\
            &\quad +2\lambda^2 q^2 (G(\frac{1}{2} +q) - \frac{1}{2})^2 + \frac{1}{2} \lambda^2 q^2 \\
            &\qquad +4 (1-\lambda)^2 \int_0^{1/2}\left(\frac{1}{2}-u\right)^2G(u)g(u)\,du + 2q^2(1-\lambda)^2 \left(G(\frac{1}{2} + q) - \frac{1}{2}\right)^2 \\
            &\quad\qquad +4\lambda(1 -\lambda)q^2 \left(G(\frac{1}{2} + q) - \frac{1}{2}\right)^2 \\
            &\stackrel{(a)}{=} 4 (1-\lambda)^2 \int_0^{1/2}\left(\frac{1}{2}-u\right)^2G(u)g(u)\,du + \frac{1 + \lambda^2}{2} q^2 \\
            &\quad + \left(-2q^2 + 2\lambda^2q^2 + 2q^2(1-\lambda)^2 + 4\lambda(1-\lambda)q^2\right) \left(G(\frac{1}{2} + q) - \frac{1}{2}\right)^2 \\
            &= 4 (1-\lambda)^2 \int_0^{1/2}\left(\frac{1}{2}-u\right)^2G(u)g(u)\,du + \frac{1 + \lambda^2}{2} q^2 \\
            &\quad + 2q^2\left(-1 + (\lambda + (1-\lambda))^2\right) \left(G(\frac{1}{2} + q) - \frac{1}{2}\right)^2 \\
            &= 4 (1-\lambda)^2 \int_0^{1/2}\left(\frac{1}{2}-u\right)^2G(u)g(u)\,du + \frac{1 + \lambda^2}{2} q^2.
        \end{aligned}
    \end{equation}
    Step $(a)$ is just a rearrangement of terms.
    This proves the upper bound for every possible value $q\in[0,1/2]$ of $\Qcal_t$.
    Taking expectations on both sides completes the proof.
\end{proof}

\subsection{Proof of Theorem~\ref{thm:centered-second-moment-general-g}}
\begin{proof}
    Since $G(u)\le\frac{1}{2}$ for $u\in[0,\frac{1}{2}]$,
    \begin{equation}
        \begin{aligned}
            4 \int_0^{1/2}\left(\frac{1}{2}-u\right)^2G(u)g(u)\,du &\leq 2 \int_0^{1/2}\left(\frac{1}{2}-u\right)^2g(u)\,du \\
            &= \Var(U)
        \end{aligned}
    \end{equation}
    by symmetry around $1/2$.

    Let $c=(1+\lambda^2)/2$ and $A_t=\Exp[\Qcal_t^2]$.
    Lemma~\ref{lem:centered-second-moment-general-g} and the preceding bound give
    \[
    A_{t+1}
    \le
    (1-\lambda)^2\Var(U)+cA_t.
    \]
    Iterating the recursion and using $A_0\le1/4$ gives
    \[
    \begin{aligned}
    A_t
    &\le
    c^tA_0+(1-\lambda)^2\Var(U)\sum_{j=0}^{t-1}c^j\\
    &\le
    \frac14c^t
    +
    (1-\lambda)^2\Var(U)\frac{1-c^t}{1-c}\\
    &=
    \frac14\left(\frac{1+\lambda^2}{2}\right)^t
    +
    \frac{2(1-\lambda)}{1+\lambda}\Var(U)
    \left(1-\left(\frac{1+\lambda^2}{2}\right)^t\right),
    \end{aligned}
    \]
    because $1-c=(1-\lambda)(1+\lambda)/2$.
    Taking $t\to\infty$ gives the stationary bound.
\end{proof}

\section{Simulation setups and more results}
\label{app:more_simulations}
Here we provide details of our simulations and more results.
We discretize the Markov chain by partitioning the interval $[0, 1]$ into $5000$ equal bins, yielding $5000$ representative nodes in the discretized chain.
For each value of $\lambda$, we construct the discretized transition matrix from the conditional CDF of the one-step update: the probability of moving from node $x_i$ to bin $[a_j,a_{j+1}]$ is
\[
    \Pr[O_{t+1}\in [a_j,a_{j+1}]\mid O_t=x_i]
    =
    F_{\lambda,G}(a_{j+1}\mid x_i)-F_{\lambda,G}(a_j\mid x_i),
\]
where $F_{\lambda,G}(\cdot\mid x_i)$ is the conditional CDF induced by drawing two independent voters from the corresponding Beta distribution, applying the anchored update with parameter $\lambda$, and taking the median with the current outcome $x_i$.
We evaluate $100$ equally spaced anchoring parameters $\lambda\in\{0.00,0.01,\ldots,0.99\}$.

We compute the stationary distribution by power iteration, using the mixing time as the stopping criterion: following Theorem~\ref{thm:wasserstein-mixing-time}, we set the mixing-time tolerance to $\varepsilon_{\mathrm{mix}} = \frac{1}{5000}$.
More precisely, for each $\lambda$ we start from the uniform distribution on the $5000$ nodes and iterate the transition matrix for the mixing-time upper bound
$\left\lceil \frac{\log(1/\varepsilon_{\mathrm{mix}})} {\log(2/(1+\lambda))} \right\rceil$
steps; the resulting distribution is used as the stationary approximation.
The empirical mixing time reported in the figures is then computed from the same trajectory as the first time $t$ for which the discretized $1$-Wasserstein distance from the time-$t$ distribution to this stationary approximation is at most $\varepsilon_{\mathrm{mix}}$.
For two binned distributions $\mu$ and $\nu$ on grid midpoints $x_1,\ldots,x_n$, this distance is computed as
$
    W_1^{(n)}(\mu,\nu)
    =
    \sum_{i=1}^{n-1}
    \left|
    \sum_{j=1}^{i}(\mu_j-\nu_j)
    \right|
    (x_{i+1}-x_i).
$
When computing the theoretical lower bound according to Theorem~\ref{thm:wasserstein-mixing-time}, we take its value to be $0$ at $\lambda = 0$.

The additional Beta distributions below show the same qualitative pattern as in the main text: stronger anchoring increases the empirical mixing time, while the stationary distributions become more concentrated around the deliberative fixed point.

\begin{figure}[t]
    \centering
    \begin{tabular}{@{}>{\centering\arraybackslash}m{0.45\textwidth} >{\centering\arraybackslash}m{0.45\textwidth}@{}}
        \includegraphics[width=\linewidth]{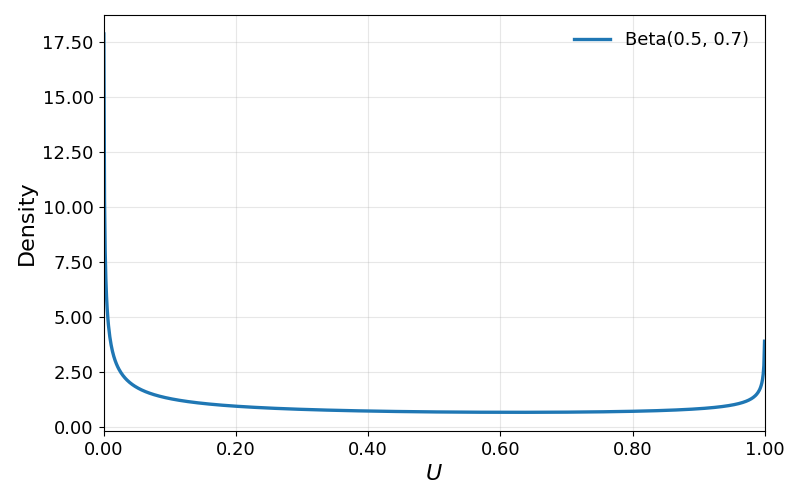} &
        \includegraphics[width=\linewidth]{fig/beta_alpha_05_beta_07/mixing_time_vs_theory.png} \\[0.3em]
        {\small (a)} &
        {\small (b)} \\[1em]
        \includegraphics[width=\linewidth]{fig/beta_alpha_05_beta_07/median_fixed_point_vs_lambda.png} &
        \includegraphics[width=\linewidth]{fig/beta_alpha_05_beta_07/stationary_density_lambdas} \\[0.3em]
        {\small (c)} &
        {\small (d)} \\
    \end{tabular}
    \caption{Combined results on $\mathrm{Beta}(0.5, 0.7)$. (a) The distribution of $\mathrm{Beta}(0.5, 0.7)$. (b) The theoretical mixing time and the estimated mixing time, we use the logarithm scale for the mixing time. (c) The population median, stationary median, and deliberative fixed point for different values of $\lambda$. (d) Larger anchoring concentrates the stationary distribution.}
    \label{fig:hist_lambda_alpha_05_beta_07}
\end{figure}

\begin{figure}[t]
    \centering
    \begin{tabular}{@{}>{\centering\arraybackslash}m{0.45\textwidth} >{\centering\arraybackslash}m{0.45\textwidth}@{}}
        \includegraphics[width=\linewidth]{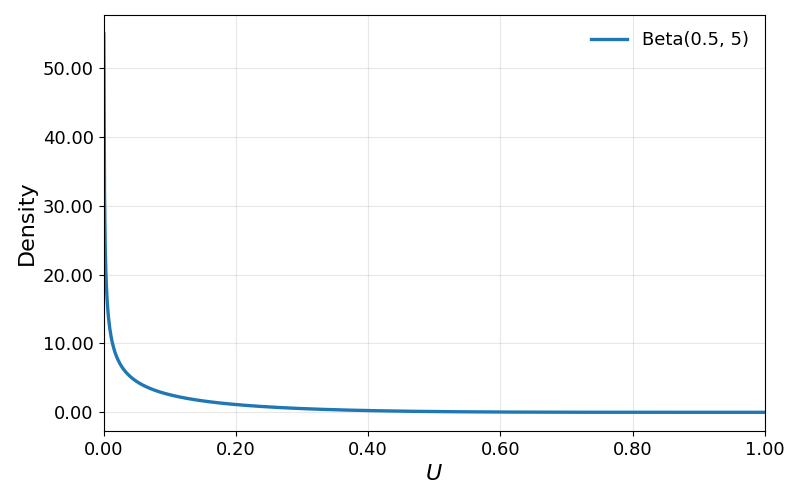} &
        \includegraphics[width=\linewidth]{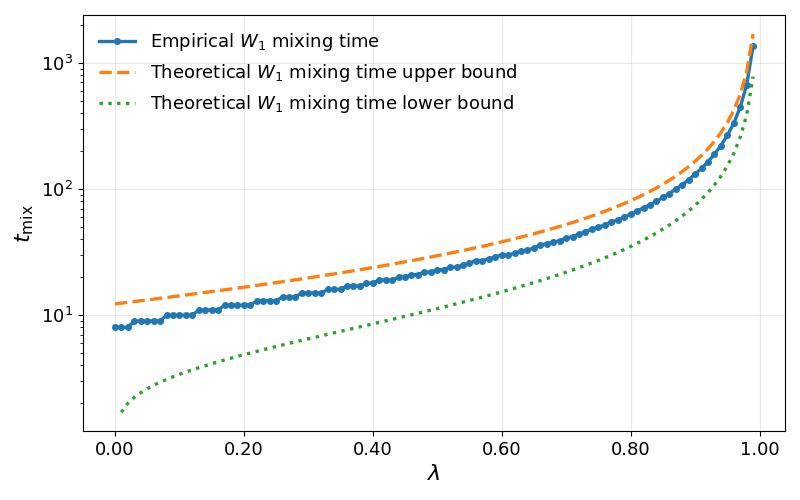} \\[0.3em]
        {\small (a)} &
        {\small (b)} \\[1em]
        \includegraphics[width=\linewidth]{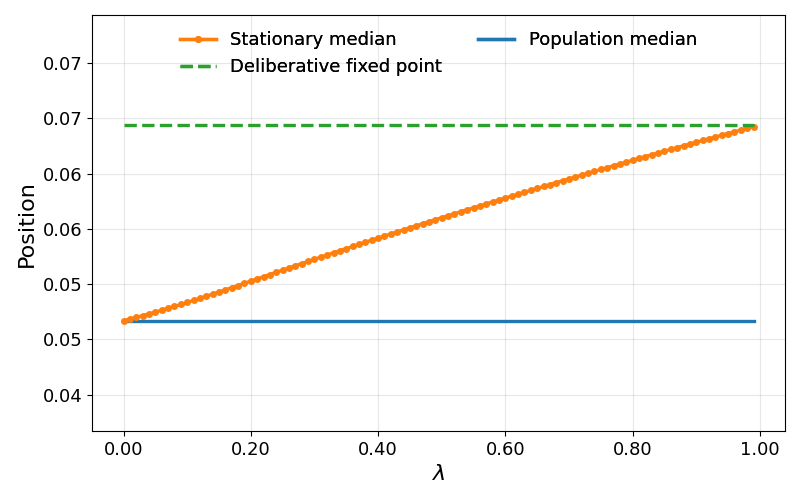} &
        \includegraphics[width=\linewidth]{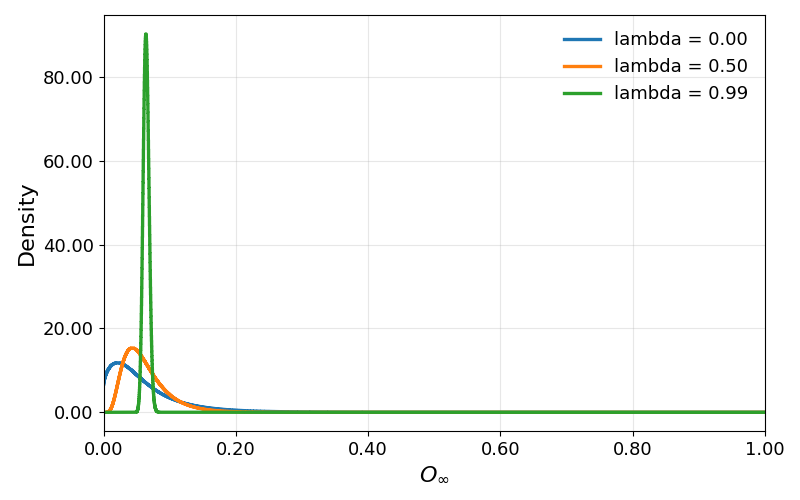} \\[0.3em]
        {\small (c)} &
        {\small (d)} \\
    \end{tabular}
    \caption{Combined results on $\mathrm{Beta}(0.5, 5)$. (a) The distribution of $\mathrm{Beta}(0.5, 5)$. (b) The theoretical mixing time and the estimated mixing time, we use the logarithm scale for the mixing time. (c) The population median, stationary median, and deliberative fixed point for different values of $\lambda$. (d) Larger anchoring concentrates the stationary distribution.}
    \label{fig:hist_lambda_alpha_05_beta_5}
\end{figure}

\begin{figure}[t]
    \centering
    \begin{tabular}{@{}>{\centering\arraybackslash}m{0.45\textwidth} >{\centering\arraybackslash}m{0.45\textwidth}@{}}
        \includegraphics[width=\linewidth]{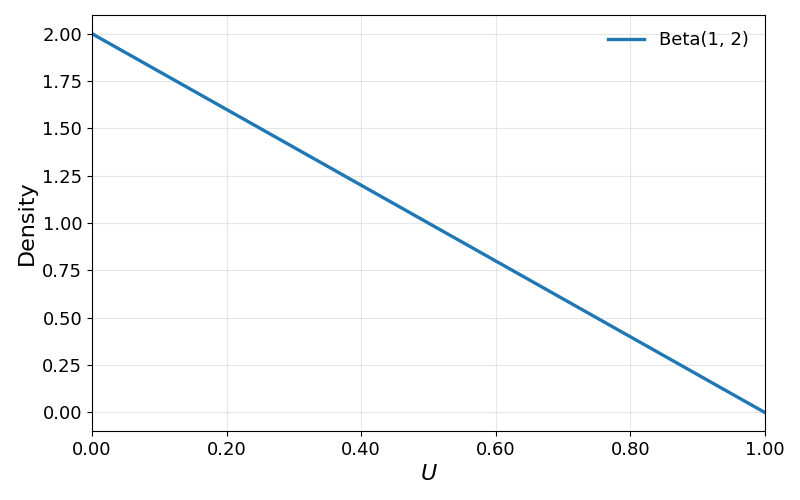} &
        \includegraphics[width=\linewidth]{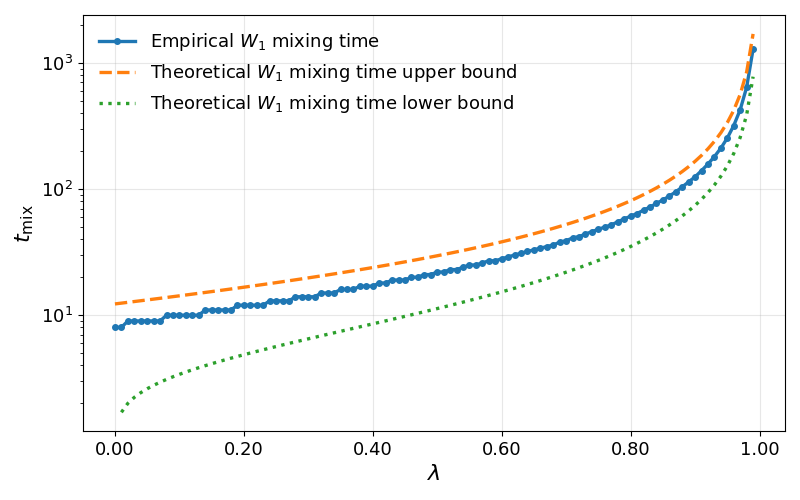} \\[0.3em]
        {\small (a)} &
        {\small (b)} \\[1em]
        \includegraphics[width=\linewidth]{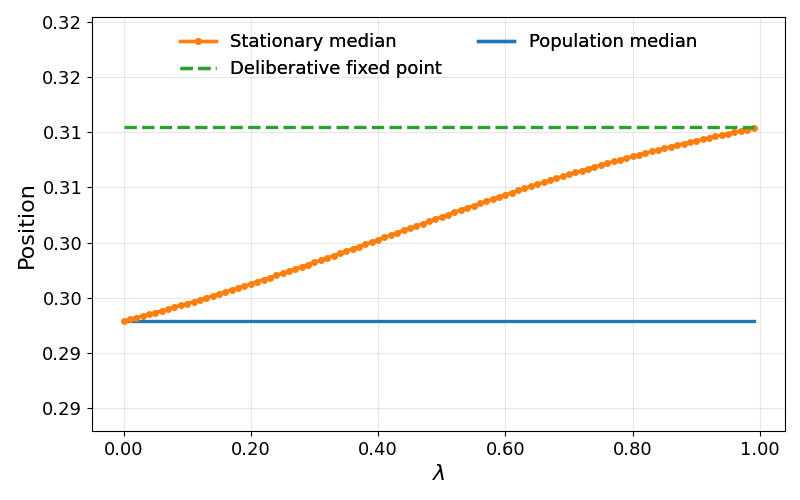} &
        \includegraphics[width=\linewidth]{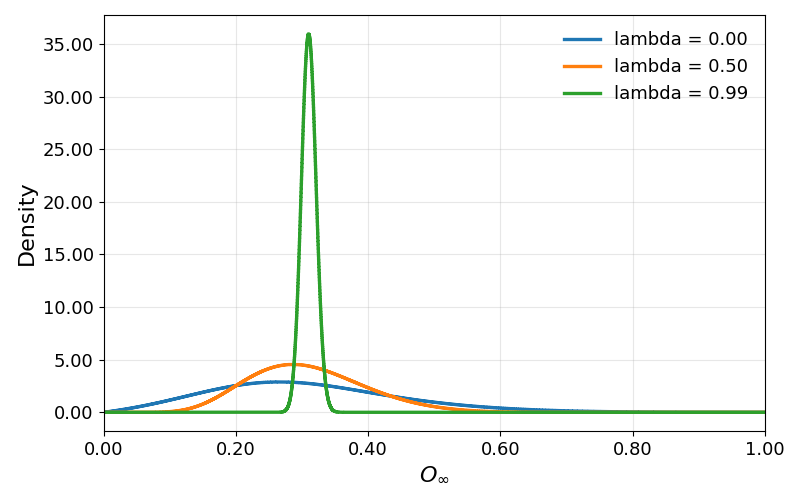} \\[0.3em]
        {\small (c)} &
        {\small (d)} \\
    \end{tabular}
    \caption{Combined results on $\mathrm{Beta}(1, 2)$. (a) The distribution of $\mathrm{Beta}(1, 2)$. (b) The theoretical mixing time and the estimated mixing time, we use the logarithm scale for the mixing time. (c) The population median, stationary median, and deliberative fixed point for different values of $\lambda$. (d) Larger anchoring concentrates the stationary distribution.}
    \label{fig:hist_lambda_alpha_1_beta_2}
\end{figure}

\begin{figure}[t]
    \centering
    \begin{tabular}{@{}>{\centering\arraybackslash}m{0.45\textwidth} >{\centering\arraybackslash}m{0.45\textwidth}@{}}
        \includegraphics[width=\linewidth]{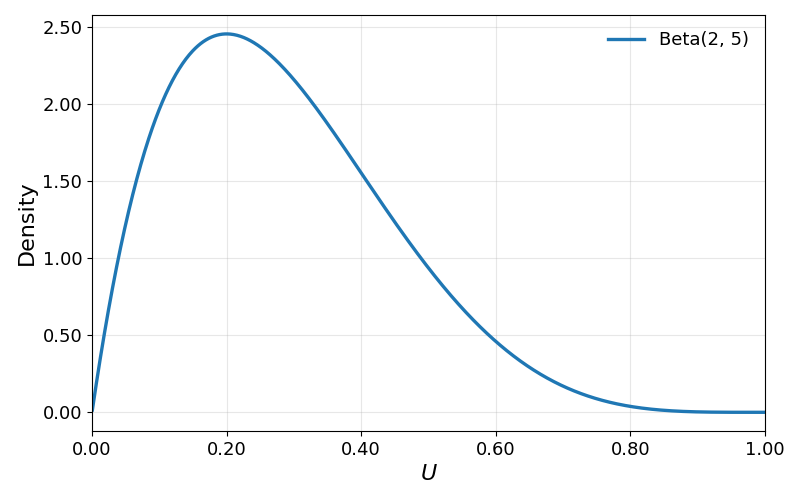} &
        \includegraphics[width=\linewidth]{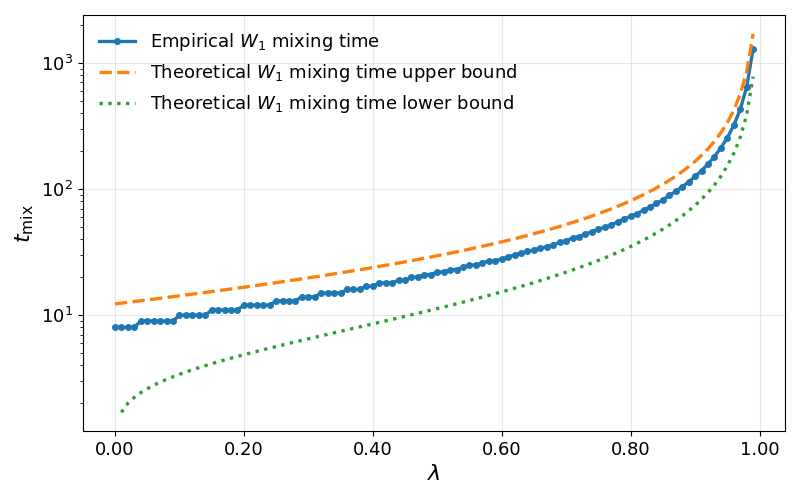} \\[0.3em]
        {\small (a)} &
        {\small (b)} \\[1em]
        \includegraphics[width=\linewidth]{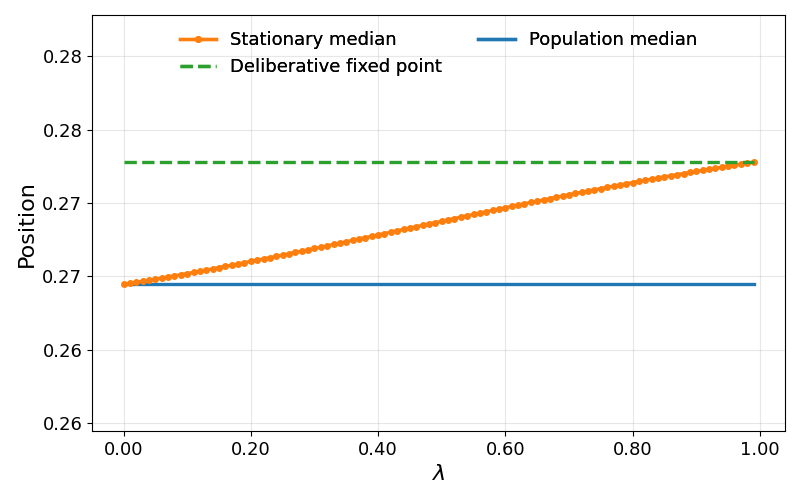} &
        \includegraphics[width=\linewidth]{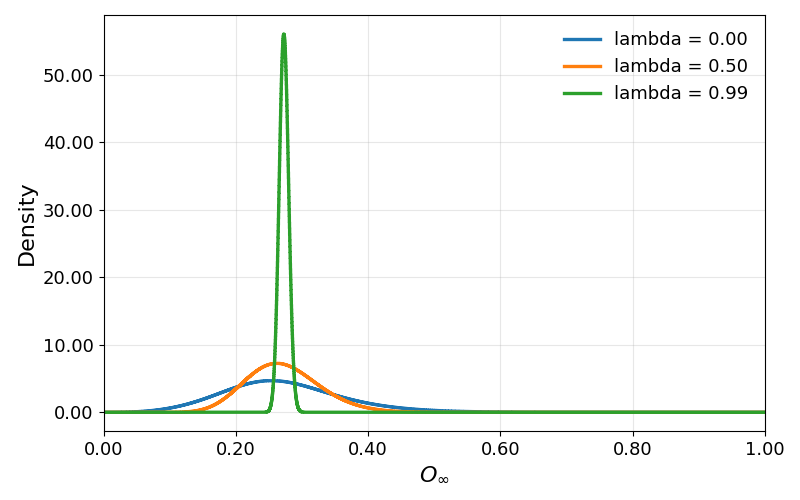} \\[0.3em]
        {\small (c)} &
        {\small (d)} \\
    \end{tabular}
    \caption{Combined results on $\mathrm{Beta}(2, 5)$. (a) The distribution of $\mathrm{Beta}(2, 5)$. (b) The theoretical mixing time and the estimated mixing time, we use the logarithm scale for the mixing time. (c) The population median, stationary median, and deliberative fixed point for different values of $\lambda$. (d) Larger anchoring concentrates the stationary distribution.}
    \label{fig:hist_lambda_alpha_2_beta_5}
\end{figure}

\begin{figure}[t]
    \centering
    \begin{tabular}{@{}>{\centering\arraybackslash}m{0.45\textwidth} >{\centering\arraybackslash}m{0.45\textwidth}@{}}
        \includegraphics[width=\linewidth]{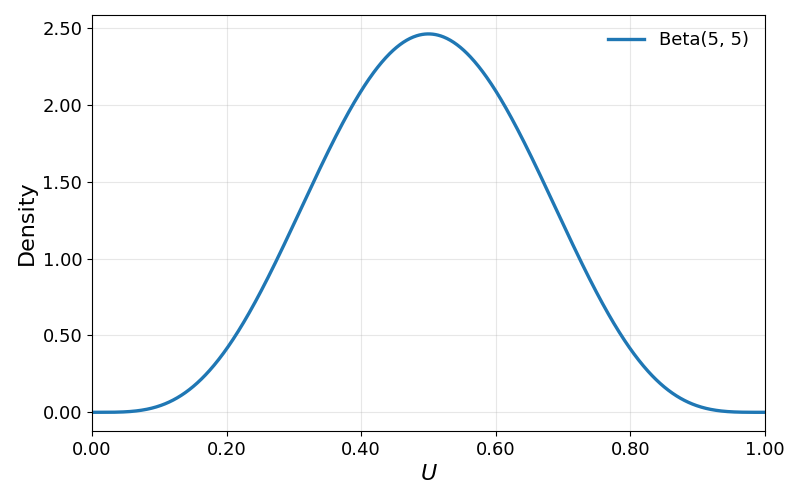} &
        \includegraphics[width=\linewidth]{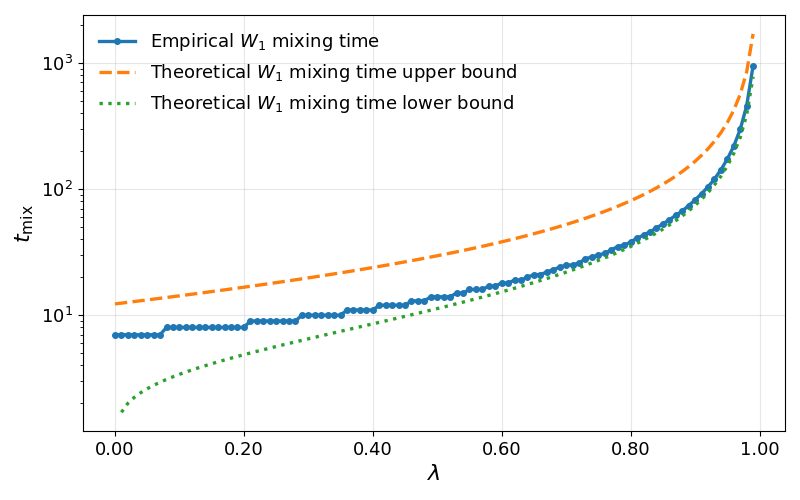} \\[0.3em]
        {\small (a)} &
        {\small (b)} \\[1em]
        \includegraphics[width=\linewidth]{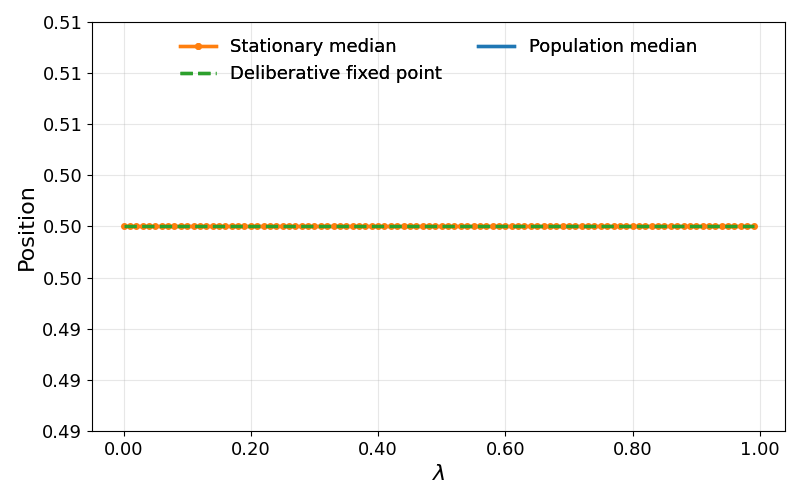} &
        \includegraphics[width=\linewidth]{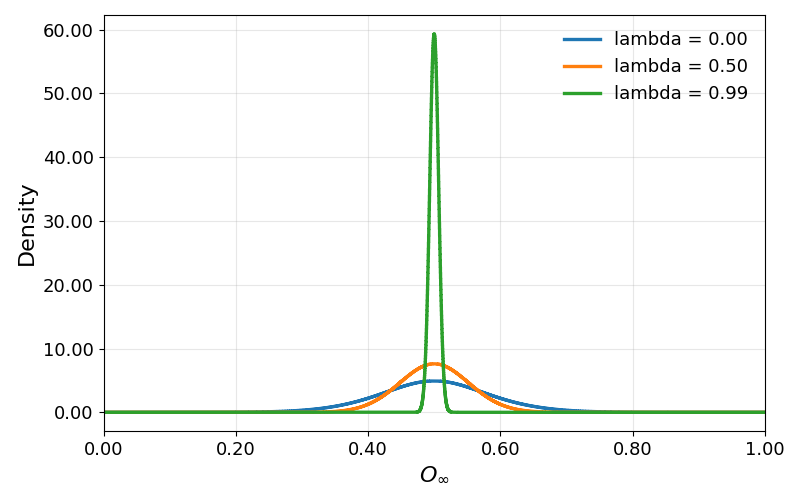} \\[0.3em]
        {\small (c)} &
        {\small (d)} \\
    \end{tabular}
    \caption{Combined results on $\mathrm{Beta}(5, 5)$. (a) The distribution of $\mathrm{Beta}(5, 5)$. (b) The theoretical mixing time and the estimated mixing time, we use the logarithm scale for the mixing time. (c) The population median, stationary median, and deliberative fixed point for different values of $\lambda$. (d) Larger anchoring concentrates the stationary distribution.}
    \label{fig:hist_lambda_alpha_5_beta_5}
\end{figure}

\end{document}